\documentclass[11pt, a4paper]{article}
\usepackage[left=2.5cm, right=2.5cm, top=2.5cm, bottom=2.5cm]{geometry}
\usepackage[utf8]{inputenc}
\usepackage{subcaption}
\usepackage{graphicx}
\usepackage[labelfont=bf,labelsep=period]{caption}
\usepackage{xcolor}
\usepackage{amsthm}
\usepackage{xspace}

\usepackage{mathdots}
\usepackage{yhmath}
\usepackage{cancel}
\usepackage{color}
\usepackage{siunitx}
\usepackage{array}
\usepackage{multirow}
\usepackage{amssymb}
\usepackage{gensymb}
\usepackage{tabularx}
\usepackage{extarrows}
\usepackage{booktabs}
\usepackage{authblk}
\usepackage{hyperref}

\usepackage{amsmath}
\usepackage{amsfonts}

\usepackage{float}

\usepackage{tikz}
\usepackage{textcomp}
\usepackage{gensymb}
\usepackage{graphicx}
\usepackage{subcaption}

\newcommand{\eg}{e.g.,\ }
\newcommand{\ie}{i.e.,\ }

\newcommand{\greedy}{Single-Exploratory-Greedy\xspace}
\newcommand{\tiebreak}{Tie Breaker Rule\xspace}
\newcommand{\tiebreaks}{Tie Breaker Rules\xspace}
\newcommand{\minerterminal}{Miner-Stable Terminal\xspace}
\newcommand{\networkterminal}{Network-Stable Terminal\xspace}
\newcommand{\mineventi}{Minimum-Score-Event$_i$\xspace}

\usepackage{pifont}
\newcommand{\cmark}{{\color{green}\ding{51}}}
\newcommand{\xmark}{{\color{red}\ding{55}}}

\newtheorem{theorem}{Theorem}[section]
\newtheorem{corollary}{Corollary}[theorem]
\newtheorem{lemma}[theorem]{Lemma}
\newtheorem{proposition}[theorem]{Proposition}
\newtheorem{definition}{Definition}[section]

\newcounter{listCounter}

\newcommand{\ListLengths}{\setlength{\itemsep}{0ex}\setlength{\topsep}{1ex}\setlength{\partopsep}{0ex}}

\newenvironment{NoIndentEnumerate}{\begin{list}{\arabic{listCounter}.}{
        \usecounter{listCounter}\setlength{\leftmargin}{1em}\ListLengths}}{\end{list}}

\newenvironment{mylowitemize}{\begin{list}{$\bullet$}{\setlength{\leftmargin}{1em}\ListLengths}}{\end{list}}

\title{Stability of P2P Networks Under Greedy Peering (Full Version)}
\author[1]{Lucianna Kiffer }
\author[2]{Rajmohan Rajaraman}
\affil[1]{ETH Zurich, {\tt lkiffer@ethz.ch}}
\affil[2]{Northeastern University, {\tt r.rajaraman@northeastern.edu}}

\date{}

\begin{document}

    \maketitle

    \begin{abstract}
    Historically, major cryptocurrency networks have relied on random peering choice rules for making connections in their peer-to-peer networks. Generally, these choices have good properties, particularly for open, permissionless networks. Random peering choices however do not take into account that some actors may choose to optimize who they connect to such that they are quicker to hear about information being propagated in the network. In this paper, we explore the dynamics of such greedy strategies. We study a model in which nodes select peers with the objective of minimizing their average distance to a designated subset of nodes in the network, and consider the impact of several factors including the peer selection process, degree constraints, and the size of the designated subset. The latter is particularly interesting in the context of blockchain networks as generally only a subset of nodes are miners/content producers (i.e., the propagation source for content), or are running specialized hardware that would make them higher performing. 
    
    We first analyze an idealized version of the game where each node has full knowledge of the current network and aims to select the $d$ best connections, and prove the existence of equilibria under various model assumptions. Since in reality nodes only have local knowledge based on their peers' behavior, we also study a greedy protocol which runs in rounds, with each node replacing its worst-performing edge with a new random edge. We exactly characterize stability properties of networks that evolve with this peering rule and derive regimes (based on number of nodes and degree) where stability is possible and even inevitable. We also run extensive simulations with this peering rule examining both how the network evolves and how different network parameters affect the stability properties of the network. Our findings generally show that the only stable networks that arise from greedy peering choices are low-diameter and result in disparate performance for nodes in the network. 
\end{abstract}

    \newpage
    \section{Introduction}

    The use of cryptocurrency networks has exploded in recent years, with billions of dollars being transferred daily across various platforms \cite{coinmarketcap}. As these networks grow, it becomes increasingly important to understand the factors that impact their structure and stability. Traditionally, networks like Bitcoin and Ethereum\footnote{Historically the two largest by market cap \cite{coinmarketcap}.} have implemented peer selection strategies which aim at randomizing peer choices to avoid attacks \cite{henningsen2019eclipsing,marcus2018low,heilman2015eclipse}, and favor properties such as low diameter \cite{bitcoincore,libp2p} for fast information propagation. Recently in Ethereum's new Proof-of-Stake network, the Beacon network, some clients have taken a new approach of greedily choosing peers based on performance \cite{Ethereum-lighthouse} where nodes drop some fraction of worst-behaving peers. In this work, we analyze a simplified model of such greedy peering strategies. 

    This paper explores the dynamics of greedy peering strategies, specifically the case where nodes aim to optimize their distance to all nodes in the network or to a special subset of high-value nodes (which we term the \textit{miners}). We are interested in characterizing how heterogeneous peering preferences affect the evolution and stability of the network. Our study is driven  by the following questions:
    


    \begin{NoIndentEnumerate}
        \item Do greedy peering strategy games have stable networks and can we characterize the conditions for stability? 
        \item What are the properties of stable networks?  When stability is reachable, how do networks converge to a stable state? 
        \item Do protocol parameters (\eg degree constraints, number of miners) alleviate or exacerbate stability dynamics?
    \end{NoIndentEnumerate}

    
    
    \vspace*{-4mm}
    
    \subsection{Our results}
        \vspace*{-.7mm}
        We consider networks of $n$ nodes where $m \leq n$ nodes are special and are labeled as \textit{miners}; these special nodes represent miners in a cryptocurrency network or content generators in a gossip-based dissemination network.  The primary objective for each node is to be ``close" to the miners in the network; we assign a score for each node as the average hop-distance to the miners in the \emph{undirected network} formed by all the edges.  The P2P network topologies that evolve in the process depend on a number of factors, including the fraction of the network that are miners, and nodes' out- and in-degree constraints.
        
        We begin by evaluating an idealized version of a non-cooperative game where nodes choose their optimal peers to minimize their own score (\ie average distance to the miners). 
        \begin{mylowitemize}
            \item {\bf Existence of stable networks:} We show that \textit{there always exist pure Nash equilibria (or stable networks) when nodes have unbounded in-degree}.  We also prove that \textit{Nash equilibria always exist for a non-trivial infinite family of networks with bounded in-degree}; the general case with bounded in-degrees is open.  These results are presented in Section~\ref{sec:ideal}.
        \end{mylowitemize}
        While the idealized game provides a sense of stable network topologies that can arise, the game is hard to implement in practice since nodes do not have global knowledge of the network and cannot make arbitrary peering choices.

        Our main results concern a natural peering protocol, which we call \textit{\greedy}, where nodes have knowledge of only their neighbors' scores and can add an exploratory random link to replace an existing link.  For this protocol, we consider a notion of stability in the network with respect to stable edges (edges that will never be dropped).  
        \begin{mylowitemize}
            \item 
        {\bf Conditions for stability:} We find that under the \greedy protocol, \textit{often the only stable networks involve a centralized, well-connected core of nodes (the miners), with all other nodes directly connecting to this core}. We also show that \textit{when nodes are faced with ties on which connection to drop, the \tiebreak is a determining factor on whether any stable topologies exist}; for instance, we show the impossibility of stability with two such rules.  Our detailed analysis of \greedy is in Section~\ref{sec:greedy} with a summary of results in Table~\ref{tab:stab-results}. 
        \item {\bf Reaching stability:} We also show that when the number of miners is capped (to at most the node out-degree capacity), if a stable network exists, networks starting from any initial state will \textit{always eventually} reach a stable state.
        \end{mylowitemize}
        
        The bulk of our theoretical results are for networks with unbounded in-degree. To explore how and if bounded networks differ in their evolution to stable topologies, we run simulations of the \greedy protocol in diverse scenarios (see Section~\ref{sec:simulations}).  
        \begin{mylowitemize}
            \item {\bf Connectivity:} Our simulations largely support that in networks where a subset of the nodes are miners, even with bounded in-degree, the miners become more connected to one another and more centralized in the network over time (\ie lower miner diameter and eccentricity), and that this behavior is quicker in networks where miners have a higher probability of connecting (\eg smaller networks, larger inbound capacity, more than one exploratory edge). 
            \item {\bf Fairness:} Additionally, across the board, our simulations show that though \greedy lowers network properties like average distance to miners, diameter, and average eccentricity, this often comes at the price of fairness, where some nodes are worse off in the network (\eg have a higher average distance to miners than in the random network) and that the network-level average advantage is likely due to a few nodes 
        being much better off and skewing the average.
        \end{mylowitemize}

    \section{Related Work}
We aim to understand the topological impact of a greedy peering protocol on a peer-to-peer network with a subset of preferential nodes everyone is trying to get close to. Most major existing P2P networks of this kind, namely cryptocurrencies, implement peering protocols that approximate low-diameter random networks \cite{bitcoincore,libp2p}. Topological analysis of these networks, however, is hard as they tend to obfuscate peer connections to avoid attacks \cite{henningsen2019eclipsing,marcus2018low,heilman2015eclipse}, and existing measurement techniques are prohibitively expensive or invasive \cite{delgado2019txprobe, li2021toposhot}, thus they tend to be run on test networks as a proof of concept. In \cite{delgado2019txprobe}, Delgado et al. map the topology of the Bitcoin Ropsten test network and compared it with simulated 
random networks of the same size and similar degrees. They find the Ropsten network to have some behavior similar to networks in our simulations (e.g., a small number of concentrated central nodes).

An inspiration for the \greedy of this paper is the Perigee protocol of Mao et al.~\cite{mao2020perigee}, which also implements a greedy peering protocol. 
The authors consider a network of all miners in an Euclidean space and show theoretically that nodes choosing neighbors geographically close to them results in network distance that is close to the optimum distance (shortest path is a constant away from the Euclidean distance). 
They show via simulations that the Perigee protocol approximates Euclidean distances between all nodes. In their simulations, they show that as processing time at each node increases, the performance of Perigee and the random network converge. In our work, this is the regime we are interested in studying: negligible link latency in comparison to node processing latency such that the properties of the overlay topology (and hop distance) are more important than the underlying communication link latency. We also focus on networks where a subset of the nodes are miners, as this is often the case \cite{miller2015discovering,kiffer2021under}. Node processing times and congestion have also been a recent focus in security analyses of blockchain consensus protocols \cite{bwlimitedposlc,kiffer2023security}. 

There is a growing body of work on peering protocols in cryptocurrency networks. Recent work has shown that variations of the Perigee protocol run by a single entity can work well as a strategy for reducing a node's own latency to the source of transactions in a network \cite{tang2022strategic,babel2022strategic}. In \cite{park2019nodes}, Park et al.\ use their own measurements of Bitcoin to propose an IP-based distance peer selection protocol to reduce round trip time of peers. Though this is in the same spirit of Perigee-like greedy protocols, using IP location is notoriously unreliable \cite{poese2011ip,gouel2021ip}, and tricky as a basis for routing as the Internet does not satisfy the triangle inequality \cite{lumezanu2009triangle}.  Other proposals for overlay networks include randomly replacing edges to maintain approximate sparse random graphs \cite{zich2008jumpnet}, and more structured proposals such as one based on the DHT Kademlia  \cite{rohrer2019kadcast}, and a hypercubic overlay network \cite{aradhya2022overchain}, as well as others \cite{toshniwal2021comparative}.

Our game-theoretic formulation falls under the category of network
creation games that have been extensively studied by researchers in
computer science, game theory, and economics.  One of the earliest
works in this space is due to Jackson and Wolinsky who introduced a
model in which there is an underlying cost for each link, and studied
tradeoffs between stability and efficiency~\cite{jackson+w:strategic}.
Demaine et al studied the price of anarchy in network creation,
quantifying the overhead of the worst stable network, when compared
with the social optimum~\cite{demaine+hmz:network}.  Many different
models based on cost, direction of links, degree constraints, and
initial conditions have been introduced, with the aim of capturing
diverse scenarios including social networks and the
Internet~\cite{BalaGoyal,corbo05,LAOUTARIS20141266,meirom+mo:network}.
Our game-theoretic model is different from models studied in past work
in the following way.  Previous work has considered either directed
networks with the node actions being given by directed edges
(e.g.,~\cite{LAOUTARIS20141266}) or undirected networks with the node
actions being given by undirected edges (e.g.,~\cite{meirom+mo:network}).  In
contrast, P2P networks underlying blockchains are better modeled by processes in which nodes select directed out-going edges to peers under some capacity constraints, while the latency incurred in the resulting network is independent of the edge direction.
Previous models in network creation games do not consider this important distinction on how links are created and their impact on delays.  As a result, the stability and convergence properties of past models differ from the results we derive for our model.  Furthermore, past work on network creation games is primarily on the nature of Nash equilibria; while we also study equilibria for our model, our focus is on the analysis of a greedy peering process that captures the peering approaches of real P2P networks.

    \section{Model}
    \label{sec:model}


We consider a network with $n$ nodes, the first $m$ of which are miners, with each node having $d$ out-going edges and at most $d_{in}\leq n$ incoming edges. Let $G$ denote the graph consisting of the nodes and the directed edges.  We use $(i,j)$ to denote the directed edge from $i$ to $j$.  For convenience, we also use $e_{i,j}$ as the indicator variable for whether either of edges $(i,j)$ or $(j,i)$ exist; so, $e_{i,j}$ is $0$ (resp., 1) if neither $(i,j)$ nor $(j,i)$ exist (resp., if $(i,j)$ exists or $(j,i)$ exists) in $G$.  Thus, $e_{i,j} = e_{j,i}$. We use the terms peer and neighbor interchangeably, and edge and connection interchangeably.  We also refer to $d_{in}$ as the inbound connection capacity.

The \textbf{\greedy} protocol operation proceeds in rounds, each of which consists of the following steps:
    \begin{NoIndentEnumerate}
        \item In an arbitrary order, each node $i$ selects at random a peer $j$ s.t. currently $e_{i,j}=0$ (thus preventing bi-directional edges) and $j$ has not reached it's inbound capacity, and adds $(i,j)$ to $G$ (thus setting $e_{i,j}=1$). If $i$ has an edge to or from every node, it adds no new edges.
    
        \item Every node $i$ updates its score $score_G(i)$ as follows. 
 If $dist_G(i,j)$ denotes the shortest distance between $i$ and $j$ in the \textit{undirected} version of $G$, then
            $$score_G(i)=\sum_{j\in miners} \frac{dist_G(i,j)}{m},$$
            which is the average distance to all miners. A miner's distance to themselves is $0$, thus the minimum score for a miner is $\frac{m-1}{m}$ and the minimum score of a non-miner is 1.
        \item Each node $i$ with $d$ out-going edges after the above steps, drops an out-going edge $(i,j)$ (i.e., sets $e_{i,j} = 0$), where $j$ is the out-neighbor of $i$ with the maximum score; if $i$ does not have $d$ out-going edges then it does not drop a peer. If two or more peers have the highest score, then a \tiebreak as defined below is used to pick one peer.
    \end{NoIndentEnumerate}

    \begin{definition} [\tiebreak]
        We define the \tiebreak as the rule used for deciding ties when a node has multiple peers sharing the highest score in step (3). We consider nodes deciding ties according to the following possible rules: uniformly at random (Random), First-In-First-Out (FIFO), Last-In-First-Out (LIFO), or a global ordering where all nodes are ranked in some order that does not change during the protocol and the node with highest rank wins (Global-Ordering).
    \end{definition}

    \begin{definition} [Stability]
            We define an edge as \textbf{stable} at some round $r$, if its value at the end of the round (after step 3), does not change for all future rounds. Formally, let $e^r_{i,j}$ be the value of $e_{i,j}$ at the end of round $r$. Edge $e_{i,j}$ is stable at round $r$, if $\forall r'\geq r$, $e^{r'}_{i,j}=e^{r}_{i,j}$.
         
         A \textbf{network-stable topology} is one where all nodes have $d-1$ \textit{stable outgoing edges}. Formally, $G$ is network-stable at round $r$ if $\forall i,j \leq n, r'\geq r$ $e^{r'}_{i,j}=e^{r}_{i,j}$. A network-stable topology is therefore an equilibrium since each node has one outgoing edge they keep replacing with a random node, but the choice of this edge will never cause a stable edge to be dropped. 
         
         A \textbf{miner-stable topology} is one where all edges between miners are stable, and therefore no further edge between miners will arise outside of the $d$-th edge. Formally, graph is miner-stable at round $r$ if $\forall i,j \leq m, r'\geq r$ $e^{r'}_{i,j}=e^{r}_{i,j}$.
         
    \end{definition}

    Note that for a topology to be network-stable or miner-stable, we require all nodes to have at least $d-1$ outgoing edges since a node will never drop an edge if they have less than $d$. Additionally, the following event arises often in our arguments related to stability.

    \begin{definition}[Special Event]
        We define \mineventi as the event where at step (1) of some round, all miners who were not yet connected to some node $i$, become connected to $i$. At step (2) of the protocol, node $i$ thus has its minimal score ($\frac{m-1}{m}$ for miners, and $1$ for non-miners). Assuming nodes have no inbound connection capacity, there is always some probability this event happens for a given $i$ in the next round.
    \end{definition}

    This event goes to show that it is always possible for any node to come to have its minimal score in any round, thus, most of our proofs involve showing that stable connections must be to minimal scoring peers so that they are not replaced.
    \newline \newline
    \noindent \textit{Model motivation.} In our model (and in practice), nodes do not know if a peer is a miner or not, all they know is a peer’s relative performance. Thus, both miners and non-miners are trying to minimize their distance to all miners via the score function. While for simplicity, in our model, we use hop distance to score peers, in practice, network propagation delays can be used locally to score peers against each other. We also adopt a graph model where the set of connections are made in a directed fashion, while their performance is evaluated over the undirected edges. This mirrors the behavior of real P2P networks, where nodes generally have a set inbound and outbound connection cap, and will accept connections as long as they have slots (allowing new nodes who join the system to find peers) but will be more selective with their out-going edges (e.g., to prevent DoS). When it comes to information propagation however (e.g., gossiping of blocks and transactions), all peers are treated the same and thus message delay is dependent only on undirected hops.  
    Additionally, we omit a peer discovery process and instead assume nodes can sample connections at random from all nodes in the network\footnote{ In reality, information of new nodes in the network needs to propagate and, though this generally happens quickly \cite{dotan2021survey}, it is possible that propagation of information on new nodes impacts how nodes cluster  in the network.}. We also consider a greedy peering process that proceeds in rounds, though within each step there is no coordination between nodes (except the assumption that in step 1 no two nodes choose each other). These simplifying assumptions allow for an initial clean theoretical evaluation.

\section{Theoretical Analyses}
Our goal is to analyze \greedy and the underlying network formation process.  To begin, we consider an idealized game-theoretic version of \greedy where nodes have full knowledge of the graph and can choose all peers to optimize their score, and explore the equilibria of this game. We then analyze the \greedy protocol, studying the stability of the networks that arise during the protocol.

\newcommand{\game}{{\mathcal G}}
\newcommand{\score}{{score_G}}
\newcommand{\junk}[1]{}
\subsection{An idealized game-theoretic model}
\label{sec:ideal}
We consider first an idealized game-theoretic model 
in which each node wants to select their outgoing peers so as to minimize the average distance to a miner.  A natural approach is via the concept of Nash equilibrium~\cite{osborne+r:game}.  Let $d$ denote the out-degree of each node.  For a given node $v$, an action is a set of at most $d$ nodes of the network representing the out-neighbors that $v$ selects. Note that in this idealized game, nodes can choose \textit{all} new outgoing edges.  Recall that while we distinguish between in- and out-edges, the distances among nodes and the average diameter are based on the graph that is obtained by viewing each edge as undirected.  

Let $\game(n,m,d)$ (resp., $\game(n,m,d,d_{in})$) denote the \emph{uncapped} (resp., \emph{capped}) \emph{game} with $n$ nodes, $m \le n$ miners, out-degree $d$, and unbounded in-degree (resp., in-degree at most $d_{in}$).  We first show that the uncapped game always has a pure Nash equilibrium, a stable topology in which no node wants to change its out-neighbors.

\begin{lemma}[Every Uncapped Game has a Stable Network]
\label{lem:stable-nash}
For all $n, m, d$, $\game(n,m,d)$ has a Nash equilibrium.
\end{lemma}
\begin{proof}

Consider any graph $G$ with the following properties: (i) there exists a node $r$ such that every node in $G$ has an edge to $r$; (ii) every edge is directed to a miner; (iii) there does not exist any pair of nodes $u$ and $v$ such that $(u,v)$ and $(v,u)$ are both present in $G$; and (iv) if a node $v$ has less than $d$ outgoing edges, then it has an edge to or from every miner.

We claim that $G$ is a Nash equilibrium.  We compute the score for each node in $G$.  The score of miner $r$ equals $(m-1)/m$ since every other miner is at distance 1; this is the lowest score possible and no action of $r$ can improve it.  Let $u \neq r$ denote any other miner, and let $i_u$ denote the number of miners that have edges to $u$.  Since $r$ is adjacent to every other node, $u$ is at distance at most $2$ from any other miner.  Given the choices of all other nodes, a best-response action of $u$ is to set all its out-edges to miners that do not have an edge to $u$; the best score achievable equals $\max\{(d + i_u) + 2(m-1-i_u-d), (m-1)\}/m$.  This is precisely what is achieved in $G$ given the conditions: by (ii), all edges are directed to a miner; by (iii) $u$ does not put any out-edges to a miner that already has an edge to $u$; finally, by (iv), if $u$ does not use all of its $d$ edges, then it already has a score of $(m-1)/m$.

We finally consider a non-miner $v$.  By the conditions, all of the edges of a non-miner are to miners.  Its distance to $\min\{d, m\}$ miners is 1 and its distance to every other miner is 2, leading to an optimal possible score of $\max\{(d + 2(m-d))/m,1\}$.  Thus, $v$ has no incentive to change its out-edges.  We have thus shown that $G$ is a pure Nash equilibrium.
\end{proof}

We next consider the network formation game with bounded in-degree.  Since the endpoints of the out-edges are constrained by this bound, reasoning about stable networks is much more challenging. We first show that the simplest network formation games in the capped setting, the family $\game(n, m, 1, d_{in})$, have pure Nash equilibrium for all $d_{in}$ and infinitely many choices of $n$ and $m$. We sketch the graphs constructed in the following proofs in Figure~\ref{fig:proof_drawing}.

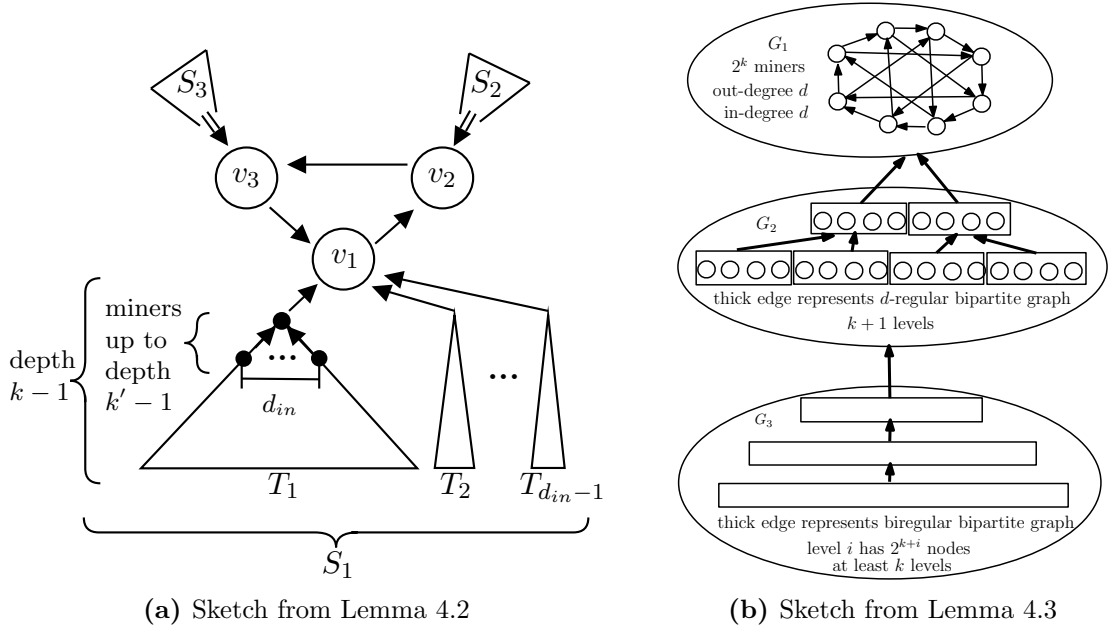
\begin{figure}
    \centering
            \begin{subfigure}{.53\textwidth}
                \centering
                \resizebox{\textwidth}{!}{
                \tikzset{every picture/.style={line width=0.75pt}} 

\begin{tikzpicture}[x=0.75pt,y=0.75pt,yscale=-1,xscale=1]

\draw   (216,213.75) .. controls (216,205.88) and (222.38,199.5) .. (230.25,199.5) .. controls (238.12,199.5) and (244.5,205.88) .. (244.5,213.75) .. controls (244.5,221.62) and (238.12,228) .. (230.25,228) .. controls (222.38,228) and (216,221.62) .. (216,213.75) -- cycle ;
\draw   (171,175.75) .. controls (171,167.88) and (177.38,161.5) .. (185.25,161.5) .. controls (193.12,161.5) and (199.5,167.88) .. (199.5,175.75) .. controls (199.5,183.62) and (193.12,190) .. (185.25,190) .. controls (177.38,190) and (171,183.62) .. (171,175.75) -- cycle ;
\draw   (262,175.75) .. controls (262,167.88) and (268.38,161.5) .. (276.25,161.5) .. controls (284.12,161.5) and (290.5,167.88) .. (290.5,175.75) .. controls (290.5,183.62) and (284.12,190) .. (276.25,190) .. controls (268.38,190) and (262,183.62) .. (262,175.75) -- cycle ;
\draw    (196,189.5) -- (212.76,204.5) ;
\draw [shift={(215,206.5)}, rotate = 221.82] [fill={rgb, 255:red, 0; green, 0; blue, 0 }  ][line width=0.08]  [draw opacity=0] (8.93,-4.29) -- (0,0) -- (8.93,4.29) -- cycle    ;
\draw    (245,206.5) -- (260.76,192.49) ;
\draw [shift={(263,190.5)}, rotate = 138.37] [fill={rgb, 255:red, 0; green, 0; blue, 0 }  ][line width=0.08]  [draw opacity=0] (8.93,-4.29) -- (0,0) -- (8.93,4.29) -- cycle    ;
\draw    (260,169.5) -- (207,169.5) ;
\draw [shift={(204,169.5)}, rotate = 360] [fill={rgb, 255:red, 0; green, 0; blue, 0 }  ][line width=0.08]  [draw opacity=0] (8.93,-4.29) -- (0,0) -- (8.93,4.29) -- cycle    ;
\draw   (201.58,241.67) -- (264.53,311.67) -- (136.21,311.67) -- cycle ;
\draw   (281.83,240) -- (291,311.5) -- (272.67,311.5) -- cycle ;
\draw   (325.33,240) -- (333,311.5) -- (317.67,311.5) -- cycle ;
\draw    (316.85,131.23) -- (288.63,117.44) ;
\draw    (316.85,131.23) -- (293.6,149.44) ;
\draw    (284.41,145.86) -- (288.63,117.44) ;
\draw    (291.69,145.72) -- (287.17,153.78)(289.08,144.25) -- (284.55,152.32) ;
\draw [shift={(281.46,160.9)}, rotate = 299.27] [fill={rgb, 255:red, 0; green, 0; blue, 0 }  ][line width=0.08]  [draw opacity=0] (8.93,-4.29) -- (0,0) -- (8.93,4.29) -- cycle    ;

\draw    (166.56,115.86) -- (140.82,133.86) ;
\draw    (163.64,151.33) -- (140.82,133.86) ;
\draw    (166.98,144.79) -- (172.66,152.63)(164.55,146.55) -- (170.23,154.39) ;
\draw [shift={(176.72,160.8)}, rotate = 234.08] [fill={rgb, 255:red, 0; green, 0; blue, 0 }  ][line width=0.08]  [draw opacity=0] (8.93,-4.29) -- (0,0) -- (8.93,4.29) -- cycle    ;
\draw    (166.56,115.86) -- (171.22,145.02) ;

\draw    (205,238) -- (216.16,227.4) ;
\draw [shift={(218.33,225.33)}, rotate = 136.47] [fill={rgb, 255:red, 0; green, 0; blue, 0 }  ][line width=0.08]  [draw opacity=0] (8.93,-4.29) -- (0,0) -- (8.93,4.29) -- cycle    ;
\draw    (325.33,238) -- (251.91,222.05) ;
\draw [shift={(249,221.33)}, rotate = 13.74] [fill={rgb, 255:red, 0; green, 0; blue, 0 }  ][line width=0.08]  [draw opacity=0] (8.93,-4.29) -- (0,0) -- (8.93,4.29) -- cycle    ;
\draw    (282,238) -- (246.13,227.58) ;
\draw [shift={(243.33,226.5)}, rotate = 21.07] [fill={rgb, 255:red, 0; green, 0; blue, 0 }  ][line width=0.08]  [draw opacity=0] (8.93,-4.29) -- (0,0) -- (8.93,4.29) -- cycle    ;
\draw   (117.33,222.67) .. controls (112.66,222.67) and (110.33,225) .. (110.33,229.67) -- (110.33,260.67) .. controls (110.33,267.34) and (108,270.67) .. (103.33,270.67) .. controls (108,270.67) and (110.33,274) .. (110.33,280.67)(110.33,277.67) -- (110.33,311.67) .. controls (110.33,316.34) and (112.66,318.67) .. (117.33,318.67) ;
\draw   (168,239) .. controls (164.16,239) and (162.24,240.92) .. (162.24,244.76) -- (162.24,244.76) .. controls (162.24,250.25) and (160.32,253) .. (156.47,253) .. controls (160.32,253) and (162.24,255.75) .. (162.24,261.24)(162.24,258.76) -- (162.24,261.24) .. controls (162.24,265.08) and (164.16,267) .. (168,267) ;
\draw  [fill={rgb, 255:red, 0; green, 0; blue, 0 }  ,fill opacity=1 ] (197.92,242.67) .. controls (197.92,240.64) and (199.56,239) .. (201.58,239) .. controls (203.61,239) and (205.25,240.64) .. (205.25,242.67) .. controls (205.25,244.69) and (203.61,246.33) .. (201.58,246.33) .. controls (199.56,246.33) and (197.92,244.69) .. (197.92,242.67) -- cycle ;
\draw  [fill={rgb, 255:red, 0; green, 0; blue, 0 }  ,fill opacity=1 ] (180.58,260.33) .. controls (180.58,258.49) and (182.08,257) .. (183.92,257) .. controls (185.76,257) and (187.25,258.49) .. (187.25,260.33) .. controls (187.25,262.17) and (185.76,263.67) .. (183.92,263.67) .. controls (182.08,263.67) and (180.58,262.17) .. (180.58,260.33) -- cycle ;
\draw    (182.67,269.67) -- (219,270) ;
\draw    (183,263.5) -- (183,274.33) ;
\draw    (219,263.5) -- (219,274.33) ;
\draw  [fill={rgb, 255:red, 0; green, 0; blue, 0 }  ,fill opacity=1 ] (215.67,260.17) .. controls (215.67,258.33) and (217.16,256.83) .. (219,256.83) .. controls (220.84,256.83) and (222.33,258.33) .. (222.33,260.17) .. controls (222.33,262.01) and (220.84,263.5) .. (219,263.5) .. controls (217.16,263.5) and (215.67,262.01) .. (215.67,260.17) -- cycle ;
\draw    (189.25,255.33) -- (196.94,247.18) ;
\draw [shift={(199,245)}, rotate = 133.34] [fill={rgb, 255:red, 0; green, 0; blue, 0 }  ][line width=0.08]  [draw opacity=0] (8.93,-4.29) -- (0,0) -- (8.93,4.29) -- cycle    ;
\draw    (217,259.17) -- (205.72,247.18) ;
\draw [shift={(203.67,245)}, rotate = 46.74] [fill={rgb, 255:red, 0; green, 0; blue, 0 }  ][line width=0.08]  [draw opacity=0] (8.93,-4.29) -- (0,0) -- (8.93,4.29) -- cycle    ;
\draw   (109.67,334.67) .. controls (109.67,339.34) and (112,341.67) .. (116.67,341.67) -- (216.83,341.67) .. controls (223.5,341.67) and (226.83,344) .. (226.83,348.67) .. controls (226.83,344) and (230.16,341.67) .. (236.83,341.67)(233.83,341.67) -- (337,341.67) .. controls (341.67,341.67) and (344,339.34) .. (344,334.67) ;

\draw (223,208) node [anchor=north west][inner sep=0.75pt]   [align=left] {$\displaystyle v_{1}$};
\draw (268,170) node [anchor=north west][inner sep=0.75pt]   [align=left] {$\displaystyle v_{2}$};
\draw (177,170) node [anchor=north west][inner sep=0.75pt]   [align=left] {$\displaystyle v_{3}$};
\draw (295.67,266) node [anchor=north west][inner sep=0.75pt]   [align=left] {\textbf{{\large ...}}};
\draw (273.67,312) node [anchor=north west][inner sep=0.75pt]   [align=left] {$\displaystyle T_{2}$};
\draw (310.67,312) node [anchor=north west][inner sep=0.75pt]   [align=left] {$\displaystyle T_{d_{in} -1}$};
\draw (193.67,312) node [anchor=north west][inner sep=0.75pt]   [align=left] {$\displaystyle T_{1}$};
\draw (73.67,255) node [anchor=north west][inner sep=0.75pt]  [font=\footnotesize] [align=left] {depth \\$\displaystyle k-1$};
\draw (151,124) node [anchor=north west][inner sep=0.75pt]   [align=left] {$\displaystyle S_{3}$};
\draw (287.63,124) node [anchor=north west][inner sep=0.75pt]   [align=left] {$\displaystyle S_{2}$};
\draw (118,231) node [anchor=north west][inner sep=0.75pt]  [font=\footnotesize] [align=left] {miners\\up to\\depth \\$\displaystyle k'-1$};
\draw (193,257) node [anchor=north west][inner sep=0.75pt]  [font=\footnotesize] [align=left] {\textbf{{\large ...}}};
\draw (191,275.33) node [anchor=north west][inner sep=0.75pt]  [font=\footnotesize] [align=left] {$\displaystyle d_{in}$};
\draw (218.67,349.33) node [anchor=north west][inner sep=0.75pt]   [align=left] {$\displaystyle S_{1}$};

\end{tikzpicture}
                }
                \caption{Sketch from Lemma~\ref{lem:nash-capped}}
            \end{subfigure} ~~
            \begin{subfigure}{.37\textwidth}
               \centering
               \resizebox{\textwidth}{!}{
\includegraphics[width=\linewidth]{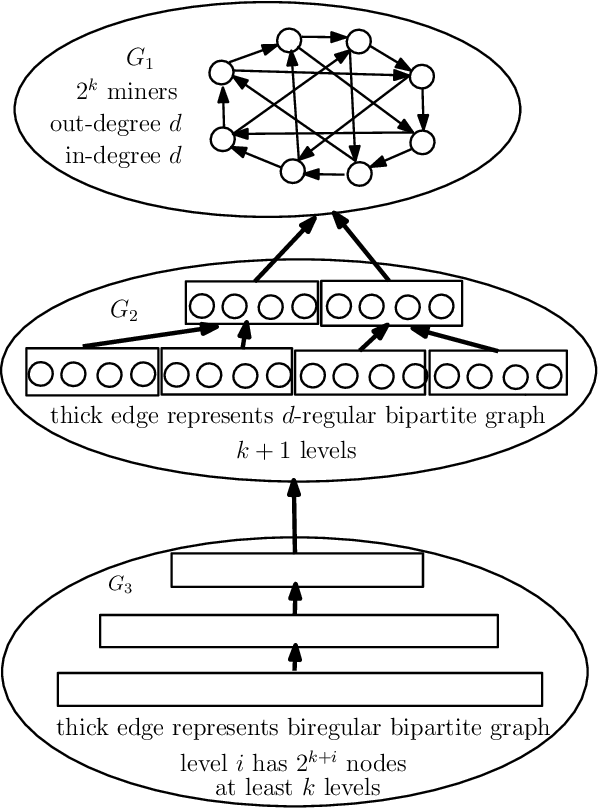}
               }
               \caption{Sketch from Lemma~\ref{lem:nash_capped_general}}
            \end{subfigure}
    \caption{{\small For the existence proofs of Lemmas~\ref{lem:nash-capped} and \ref{lem:nash_capped_general} we construct the above network families.  Both networks are essentially layered, and the crux of the proof is that the miners are positioned in the ``top layer(s)" with the non-miners occupying the remaining layers, such that all nodes except the non-miners in the last layer have full in-capacities. The nodes in the last layer have the highest score and no node higher up in the layers would give up their position to connect to a node in the bottom layer.}}\label{fig:proof_drawing}
\end{figure}

\begin{lemma}[Stable Networks in Capped Games with Unit Out-Degree]
\label{lem:nash-capped}
For any $n = 3\cdot d_{in}^k$ and $m = 3\cdot d_{in}^{k'}$ for any integers $k$, $k'$, $k \ge k' \ge 0$, $\game(n, m, 1, d_{in})$ has a pure Nash equilibrium.
\end{lemma}

 \begin{proof}
        Note that $n$ and $m$ are at least 3. Let $C$ be a directed triangle over three miners $v_1$, $v_2$, and $v_3$.  Partition the remaining nodes into three equal-sized sets $S_1$, $S_2$, and $S_3$.  By the assumption on $n$, each $S_i$ has $d_{in}^k - 1$ nodes and $d_{in}^{k'}-1$ miners.  Split each $S_i$ into $d_{in}-1$ balanced complete $d_{in}$-ary directed trees, with every edge directed to the parent.
        Note that a balanced complete $d_{in}$-ary tree of depth $h-1$ has $\frac{d_{in}^h-1}{d_{in}-1}$ nodes. Furthermore, assume that all nodes in each $d_{in}$-ary tree $T_j$ up to depth $k'-1$ are miners (there are exactly $\frac{d_{in}^{k'}-1}{d_{in}-1}$ such nodes in each $T_j$). Let $r_{i,j}$ for $1\le j \le d_{in}-1$ be the root of each tree $T_j$ in $S_i$. Consider the network $G$ consisting of the union of $C$, $S_1$, $S_2$, $S_3$, and the edges $\{(r_{i,j}, v_i): 1 \le i \le 3\}$.  We show that $G$ is a pure Nash equilibrium.
    
        For each $v_i$, consider the directed tree $Tv_i$ made up of the trees in $S_i$ and the edges from their roots to $v_i$. By symmetry, the scores of all nodes at a given depth in $Tv_i$ are the same. Let $\delta_1$ and $\delta_2$ be integers such that $k \ge \delta_1 > \delta_2$. Consider two nodes $n_1$ and $n_2$ in $S_i$ at depths $\delta_1$ and $\delta_2$, respectively.  The score of $n_1$ is strictly greater than that of $n_2$. This is because a node $v$ at some depth $\delta$ has a lower score than any of the nodes in the sub-tree below it since there are twice as many miners not in $Tv_i$ as are in $Tv_i$, and $v$ is in the only path of its sub-tree to those miners. This also implies that for any leaf $\ell$ and any non-leaf node $v$, the score of $\ell$ is at least the score of $v$. Given the symmetry of the leaves, this applies for all $n_1,n_2$ at depths $\delta_1$ and $\delta_2$ in different $Tv_i$ trees. 
    
        To complete the proof, we consider each node and verify that their out-edge in $G$ gives the node the best score, given the choices of other nodes.  By construction, every non-leaf node in $G$ has a full in-degree of $d$.  Hence, the action of any node $v$ in the game is to either have their current out-edge in $G$ or to direct the out-edge to one of the leaf nodes. If $v$ selects a leaf node $\ell$, note that by symmetry this node has score $s$ at least that of $v$. If $v$ is not a miner, then its distance to all miners will be one more than that of $\ell$, so the new score of $v$ is $s+1$. 
    
        If $v$ is a miner at depth $h$, we consider the change in its score as new\_score$(v)-$old\_score$(v)$ and prove this change is positive. Note first that $v$'s distance to itself and all miners in the sub-trees that point to $v$ remains the same, so a difference of 0. WLOG, assume $v$ was in $Tv_1$ and chose to connect to a leaf in $Tv_2$. Thus $v$'s distance to miners in $Tv_3$  has only increased by $k+1-h$.    
        Next, let $v'$ be the new ancestor of $v$ in $Tv_2$ at height $h$. Since $v$ and $v'$ were at the same height in $G$, by the symmetric property of $G$, they had the same cumulative distance to miners in $Tv_1$ and $Tv_2$ who are not in the sub-trees of $v$ or $v'$. Thus $v$'s score to these miners has increased by its distance to $v'$ which is again $k+1-h$.
        
        The only remaining miners to consider in the change in score is from $v$ to the miners in the sub-tree of $v'$ which may now be closer (i.e., take away from the change in score). In the original structure, $v$ was a maximum $2h+1$ from $v'$ ($2h$ if $v$ and $v'$ are in the same tree), thus these miners are now at most $k+h+1$ from $v$. If we consider only the sub-trees of $G$ whose roots are at depth $h$, let $X$ be the number of miners in each of these sub-trees. There are $(d-1)*d^{h-1}$ such sub-trees in each $S_i$. Then new\_score$(v)-$old\_score$(v)$ is at least the score difference for these sub-trees (since we are just under-counting miners who are now further away from $v$), we get that 
        $$ \text{new}\_\text{score}(v)-\text{old}\_\text{score}(v) \geq (k+1-h)\cdot \left[ 3 \cdot (d-1)\cdot d^{h-1} -2 \right]\cdot X  - (k+h+1)\cdot X.$$ 
        Let's assume this score difference is non-positive. Remembering that  $d\geq 2$, we can simplify the above inequality to be 
        $$2^{h-1}\leq (d-1)d^{h-1}\leq \frac{3k+h-3}{3k-3h+3}$$
        which is not satisfiable. Thus the score difference is strictly positive, meaning $v$ would not choose to change its out-edge.
    \end{proof}

We next extend Lemma~\ref{lem:nash-capped} to capped games with arbitrary out-degrees, establishing the existence of stable networks in much more general settings in Lemma~\ref{lem:nash_capped_general}.  At a very high level, the proof establishes stability of networks with a certain structure: a strongly connected core $G_1$ consisting of the miners, a hierarchical layered network $G_2$ with nodes in its top layer having directed edges to $G_1$, and a more loosely structured layered network $G_3$ with nodes in its top layer having directed edges to $G_2$ and nodes in its bottom layer having no incoming edges. 
 The nodes in the bottom layer of $G_3$ are farthest from the miners while the in-degree of every other node in the network equals the in-capacity, which ensures that the network is stable.  Figure~\ref{fig:proof_drawing}(b) illustrates this proof.
\begin{lemma}[Stable Networks in Capped Games with Arbitrary Out-Degree]
\label{lem:nash_capped_general}
For integers $k \ge 0$, $d \ge 2$, $m \le 2^k$, and $n \ge 2^{3k}$,  $\game(n, m, d, 2d)$ has a pure Nash equilibrium.
\end{lemma}

\begin{proof}
  We consider a network $G$ that has three components. The first component $G_1$ is a graph with its vertex set being all the miners, each node with an out-degree of $d$ to other miners, leaving $d2^k$ incoming capacity. 
  
  The second component $G_2$ consists of $k+1$ levels of nodes, with level $i \ge 0$ consisting of set $L_i$ of $2^{k+i}$ nodes.  For each $i$, we partition the $2^{k+i}$ nodes into $2^{i}$ groups of $2^k$ nodes each; we label these groups $S_{i,0}, \ldots, S_{i,2^{i+1}-1}$.
  For any non-negative integer $j_{k-1} < 2^{k}$, the node groups $S_{i,j_i}$, $j_i = \lfloor j_{i+1}/2 \rfloor$, $-1 \le i < k$, include a butterfly network.  That is, 
  the edges between $S_{i,j_i}$ and $S_{i+1,j_{i+1}}$, with $j_i = \lfloor j_{i+1}/2 \rfloor$ include the $i$th level of an $2^k$-input butterfly network.  This can be set up since $d \ge 2$. 
  This ensures that every node in $L_{k}$ has a path of length $k$ to every node in $S_{1, 0}$.   (Instead of using a butterfly network, we could also have random biregular bipartite graphs between consecutive levels, achieving the same distance property between $L_{k}$ and $S_{1,0}$ with high probability.) We can then place all remaining edges from each $L_i$ arbitrarily to nodes in $L_{i-1}$. We then have all edges from $S_{1,0}$ go into $G_1$, thus all nodes at level $L_{k}$ have the same score as all their paths to the miners converge at $L_0$.


 The final component $G_3$ consists of at least $k$ levels of nodes, which we number from $k$.  For level $i \ge k$, the set $L_i$ consists of $2^{k+i}$ nodes, with the last level $\ell$ possibly having fewer than $m 2^\ell$ nodes.  There is an arbitrary bipartite graph connecting $L_i$ to $L_{i-1}$ with $d$ outgoing edges for each node in level $i$ and $2d$ incoming edges for each node $L_{i-1}$.  

 We now argue that $G$ is a stable network.  We make three observations about distances, which are critical in establishing that no node benefits by changing any of its outgoing edges.
  
    First, for any $0 \le i < k$, for any node $u$ in level $i$ and miner $v$, the distance  between  $u$ and $v$ is at most $2k - i + \delta_v$, where $\delta_v$ is the minimum, over all nodes in $L_0$, of the distance between the node and $v$.  This follows from the fact that any node in $L_k$ has a shortest path to every node in $L_0$; consequently, any node in $G_2$ has a path to $v$ by concatenating a subpath of length $k-i$ to the closest node in $L_k$, followed by a subpath of length $k$ to the node in $L_0$ that is closest to $v$, followed by a subpath of length $\delta_v$ to $v$.
  
  Second, for any $i \ge k$, for any node $u$ in level $i$ and miner $v$, the distance between $u$ and $v$ is exactly equal to $i + \delta_v$.  This holds since any node in $L_i$ has a path to $v$ by concatenating a subpath of length $i-k$ to an arbitrary node in $L_k$, 
  followed by a subpath of length $k$ to the node in $L_0$ that is closest to $v$
  followed by a subpath of length $\delta_v$ to $v$.
  
Finally, we note that for any leaf node $u$ and any miner $v$, the distance between $u$ and $v$ is at least $\ell + \delta_v$.  To complete the proof, we observe that for any $u$ in $G$ and any miner $v$, the distance between $u$ and $v$ is at most $\max\{i, 2k - i\} + \delta_v \le \ell + \delta_v$; consequently, replacing any out-edge with an out-edge to a leaf will not decrease its cost.  Therefore, $G$ is a stable network.
\end{proof}

Lemma~\ref{lem:nash_capped_general} can be further extended to allow for more general relationships between the in-degree and the out-degree (e.g., when $d_{in} = cd$ for any integer constant $c \ge 3$).  We do not know, however, if stable networks exist for all choices of $n$, $m$, $d$, and $d_{in}$.  We also note that the proofs of Lemmas~\ref{lem:stable-nash} through~\ref{lem:nash_capped_general} are all existence proofs in the sense that they present \emph{specific} families of stable networks.  There are, however, many other families of stable networks, as suggested by the proofs.  In particular, the proof of Lemma~\ref{lem:nash_capped_general} presents a layered structure for stable networks that accommodates a variety of different topologies.

\subsection{\greedy}
\label{sec:greedy}
    Our results in Section~\ref{sec:ideal} indicate that there exist Nash equilibria in the full-information setting, where every node knows the current network topology at all times.
    In practice, however, nodes operate with information only about their own peers. We consider now the \greedy protocol defined in Section~\ref{sec:model} and analyze networks that evolve in this process.  Instead of equilibria, we consider miner and network stability as defined in Section~\ref{sec:model}. We focus on two questions: does there always exist a stable topology? what properties do stable topologies have? Table~\ref{tab:stab-results} summarizes our results and the impact of the \tiebreak on the ability of the network to stabilize.  We begin by defining conditions for miner stability.


    \begin{table}
        \centering
        \begin{tabular}{c|c|c|c|c}
            \hline 
             & \multicolumn{2}{c|}{\bfseries \underline{~$m<2d$~}}
            &\multicolumn{2}{c}{\bfseries \underline{~$m\geq2d$~}}\\  & miner-stable & network-stable & miner-stable & network-stable \\ \hline
            Global-Ordering/LIFO & \cmark ~Prop.~\ref{prop:miner-clique} & \cmark ~Lem.~\ref{lem:stable-top} & \cmark ~Lem.~\ref{lemma:global-stable} & \xmark ~Lem.~\ref{lem:no-stable-top}\\ \hline 
            Random/FIFO & \cmark ~Prop.~\ref{prop:miner-clique} & \xmark ~Lem.~\ref{lem:ran-no-stable}& \xmark ~Prop.~\ref{prop:imposs-ran-stable}& \xmark ~Lem.~\ref{lem:ran-no-stable} \\ \hline 
        \end{tabular}
        \caption{Summary of our results of Section~\ref{sec:greedy} for the existence and impossibility of miner-stable and network-stable topologies in the unbounded \greedy game. Note for network-stability, we assume $n\gg d$, as discussed in their proofs.}
        \label{tab:stab-results}
    \end{table}

    \begin{proposition}[Clique Miner Stability for Small $m$] 
            If $m < 2d$ there exists a miner-stable topology for any \tiebreak.
        \label{prop:miner-clique}
    \end{proposition}
    \begin{proof}
            Since $m< 2d$, there are enough edges between miners such that they can connect in a clique with each miner having at most $d-1$ outgoing edges to any other miner. Each miner $i$ has a distance $0$ to itself and a distance $1$ to each of the other miners, thus their weighted distance to miners in the clique is $w_i = \frac{m-1}{m}$.  Any node $i$ drops a connection to the peer with the largest score among the peers of $i$. Any non-miner node $h$ must have a distance of at least 1 to each miner, thus $score(h) \geq 1 > \frac{m-1}{m}$, thus a non-miner score will not tie with a miner score, therefore miners will only ever drop a peer who is not a miner. Since all miners are already connected, a miner will never establish a new connection to another miner. 
    \end{proof} 

    \begin{lemma}[Global-Ordering and LIFO Miner Stability]
            For all $m$ and $d$, there exists a miner-stable topology for the cases where the \tiebreak is Global-Ordering or LIFO. 
        \label{lemma:global-stable}
    \end{lemma}
    \begin{proof}
            For $m< 2d$, by Proposition~\ref{prop:miner-clique} a clique is miner-stable for any tie-breaking rule.         
            For all other values of $m$ and $d$, we first examine the Global-Ordering rule. Suppose the first $d-1$ miners (by global ordering) connect in a clique and all other miners have their $d-1$ edges to the first $d-1$ miners. The first $d-1$ miners are therefore connected to all miners and have a minimum score $\frac{m-1}{m}$. In a round, these $d-1$ miners will not drop one another since they have the lowest score and priority ranking, and all other miners will not drop these first $d-1$ miners for the same reason. Since the first $d-1$ miners are connected to all miners, they will never develop a new connection to a miner. In a round, two miners $i,j$ could connect to each other via their $d$-th outgoing edge. Since either miner could at best have the same score, but worst ranking as the first $d-1$ miners, this new edge will not replace the existing $d-1$ out-going edges to the first $d-1$ miners. 

            Next, consider the LIFO \tiebreak. Consider the graph where some set of $d-1$ miners connect in a clique, and all other miners have outgoing edges to these miners. Let the non-clique miners have their $d$-th edge to some non-miner. Note that the clique does not have any outgoing edges to the non-clique miners and the non-clique miners have no edges among themselves. Since all edges between miners are going into the clique, for any new edge to replace one of these edges, it would need to have an equal score but would be the newest edge, by LIFO it would be dropped. The clique is connected to all miners so it can't make a connection to any other miner. Miners not in the clique thus have $d-1$ stable edges. 
    \end{proof}

    We note that the networks presented in the proof of Lemma~\ref{lemma:global-stable}  are not stable under FIFO (or Random) as there is always a chance all miners connect to some non-clique miner and this miner would always (or with some probability) replace an existing miner in the clique. We extend this logic to the following claim. 

    \begin{proposition} [Impossibility of Miner-Stability under FIFO or Random for Large $m$] 
            If $2d\leq m$, and ties are broken by FIFO or randomly, there is no miner-stable topology. 
        \label{prop:imposs-ran-stable}
    \end{proposition}
    \begin{proof}
            Assume there is some miner-stable topology. This means all edges between miners are either stable or will always be dropped at the end of the round (\ie are the maximum-score edge after the \tiebreak). Since $2d\leq m$, miners cannot form a \textit{stable} clique, thus there is some miner $m_i$ that is not connected to all miners. There is some probability that in the next round, all miners connect to this miner with at least one miner $m_j$ having a new out-going edge to $m_i$. Miner $m_i$ now has score $\frac{m-1}{m}$; so, $m_j$ will not drop the edge to $m_i$ since it is either smaller than some existing edge or ties with all edges and by FIFO (or random) \tiebreak replaces some edge (has a probability of replacing some edge). 
    \end{proof}

    We now derive conditions under which stable networks exist and conditions where stability is impossible. For the following cases, we generally assume $n\gg d$, and $n>2d^2$ for proofs of instability\footnote{Note that this is true of Bitcoin (where $d=8$ and $n>10000$) and Ethereum ($d\approx 16$ and $n>2000$).}. Recall that since all nodes have at least $d-1$ outgoing edges at the start of the protocol, nodes always have at least $d-1$ outgoing edges during any portion of the protocol.

    \begin{lemma}[Stable Topology for Small $m$]
            For all $m< 2d$ with \tiebreaks LIFO or global ordering, there exists a stable topology.
        \label{lem:stable-top}
    \end{lemma}
    \begin{proof}
        For $m<d$, the following topology is stable: Miners connect in a clique thus all have score $\frac{m-1}{m}$. All non-miners connect to all miners, either via incoming or outgoing edges. With global ordering, we have that in the miner clique, a miner with $k$ outgoing edges to other miners has $d-1-k$ outgoing edges to the first $d-1-k$ non-miners by the global order.  Every non-miner connects to all $k$ miners it is not connected to and similarly has its remaining $d-1-k$ outgoing edges to the first $d-1-k$ (or if no global ordering, to any) of the remaining non-miners. Thus all non-miners have the same score, and no non-miner will drop an edge to a miner. We then set everyone's $d$-th edge, note miners never have a $d$-th edge since they are already connected to everyone. For non-miners, any new edge will be to a node not in the first $d-1-k$ non-miners defined above (or any non-miner for LIFO), and therefore will never win the tie-breaker.

        For $d\leq m < 2d$, we present a stable topology with all miners connecting in a clique using $(m-1)/2$ out-going edges and all remaining outgoing edges from the miners going to some (top-ranked) $n'$ non-miners who have a score of 1 and all outgoing edges into the clique. Thus the miners have a miner-stable topology and the $n'$ non-miners will keep their edges to the miners as long as the miners keep their edges to these non-miners and vice versa. All non-miners that are not part of the $n'$ non-miners connect to the top $d-1$ miners, thus the
        top $d-1$ miners have no $d$-th edge and the rest have an edge to some non-miner that will at best have a score of 1 but be higher-ranked than the $n'$ non-miners. The $m$ miners use $(m-1)/2$ outgoing edges each to connect to the clique, leaving $d-1-(m-1)/2$ to connect to the $n'$ non-miners. The $n'$ non-miners have a cumulative $n'(d-1)$ edges, that need to be enough to connect them to all $m$ miners, thus $n'm=n'(d-1)+ d-1-(m-1)/2$, yielding
        $n' = \frac{(2d-1)-m}{2}$.  Note that when $m=2d-1$ the miners use all their outgoing edges to connect in the clique so $n'=0$.
    \end{proof}

    Next, we show how all remaining cases cannot reach stability. 

    \begin{lemma}[LIFO and Global Ordering Tie-Breaking Instability for Large $m$]
            For all $m\geq 2d$, with \tiebreaks LIFO or global ordering, there exists no stable topology for $n\gg d$.
        \label{lem:no-stable-top}
    \end{lemma}

    \begin{proof}
        In this proof we essentially show that in a stable topology, 1.) there is a subset of miners who have the minimum score and all nodes must have all of their outgoing edges to this subset for the edge to be stable 2.) there are not enough miners in this subset so that all their outgoing edges are to each other, so some of their edges would be unstable.

        We know from Proposition~\ref{prop:first-d-1-stable}, any miner-stable topology must have that at least $d-1$ miners are connected to all miners (i.e., have the minimum score of $\frac{m-1}{m}$). We first argue that in a stable topology, all non-miners are also connected to some $d-1$ of these minimum-scored miners. If some non-miner is not, then in the next round it could connect to one of these miners it is not connected to and this connection will replace an existing connection (since less than $d-1$ of them have the minimum score). 
            
        To prove stability, we need to show that there does not exist an edge to a non-miner from some node which could be replaced by a new lower-scored edge to a non-miner in some round (i.e., there does not exist an unstable edge \textit{to} a non-miner). Since $n \gg d$, all non-miners cannot connect to all miners, so it must be that any stable edge from a miner to a non-miner must be to a non-miner with a score of 1 (for global ordering, these non-miners are the top-ranked ones), say there are $n'$ such non-miners. Equally, there are not enough edges for all non-miners to connect to each other. Consider a stable edge from non-miner $n_i$ to non-miner $n_j$. It must be that $n_i$ has a score of 1 (i.e., there does not exist a miner that can replace $n_j$ or another peer in the next round), and either $n_j$ has a score of 1 (otherwise a non-miner could get the lowest score in the next round and replace it) or all non-miners connect to $n_i$ (so there is no other out-going edge possible to replace $n_j$). This implies all non-miners with score $>1$ have all stable edges to miners, and thus at most 1 non-miner connects to all non-miners (with the $d$-th edge of these nodes) and this can change in any round, so there is no \textit{stable} edge to a non-miner with score $>1$. 

        Say a miner $m_i$ with score $>\frac{m-1}{m}$ has a stable edge to a miner $m_j$ with score $>\frac{m-1}{m}$, in the next round $m_i$ can connect to a miner it was not connected to before who in this round has score $\frac{m-1}{m}$ and replaces $m_j$ or other edge, thus the topology was not stable. We get that all miners with score $>\frac{m-1}{m}$ have all $d-1$ edges to miners with score $=\frac{m-1}{m}$ in a stable topology, so $n'$ is at most 1. So the single non-miner with a score $1$ would have only stable edges to miners, thus all stable edges from non-miners are to miners and all stable edges from miners are to miners. Since in some round a non-miner could have all miners connect to it, any stable edge to a miner must be to a miner with a score at most 1. 

        If there are no miners with a score of 1, then all miners with a score of $\frac{m-1}{m}$ must connect to each other with all their edges (needs to be $2d-1$ such miners), and all other miners must connect to these miners with their outgoing edges (max $d-1$ such miners). Thus, this is impossible.
            
        Say there are $m''$ miners with a score of 1, the rest must have score $\frac{m-1}{m}$ otherwise they would not connect to the score 1 miners. For a miner to have a score of 1, they must connect directly to all miners but 1, who they are 2 hops away from. So there are $m''= 2$ miners with score 1, so $m'=m-2$. For the $m'$ miners to form a clique, it must be that $m'\leq 2d-2$, thus given the bounds on $m$ we started it, we have that $2d-1< m \leq 2d+1$. For the $m'$ miners to connect to each other, we need $(m'-1)m'/2$ edges, and for the 2 miners to connect to them we need an additional $2m'$ edges, this must be equal to the number of edges of the $m$ miners thus $m(d-1)=\frac{m'(m'-1)}{2}+2m'$, substituting in $m'=m-2$ and simplifying, we get that we need
            $$(2d-2)=\frac{m'(m'+3)}{m'+2}$$
        to be satisfied. If $m=2d$, this gives us that $m'+2=m'+3$. If $m=2d+1$, we get $(m'-1)(m'+2)=m'(m'+3)$ which is only true for $m'=-1$. Thus, there are no conditions for stability.    

    \end{proof}

    \begin{lemma}[FIFO and Random Tie-Breaking Instability for Large $n$]
            For all $m,d$, and $n>2d^2$ with \tiebreaks FIFO or random, there exists no stable topology.
        \label{lem:ran-no-stable}
    \end{lemma}
    \begin{proof}
            \noindent(Case 1: $m<d$) Assume there is a stable topology, this means all edges are stable and by Proposition~\ref{prop:only-stable-clique} the miners connect in a clique. The first observation is that, in a stable topology, all non-miners would connect to all miners. To see this, first, assume there is some non-miner $n_i$ who is not connected to all miners, there is a chance that in the next round, all miners connect to this $n_i$. Either this happens with $n_i$ making a connection to at least one new miner $m$, since miners have the minimum score, this connection will replace an existing connection to a non-miner. Or, this happens with some miner $m$ making a connection to $n_i$, in this round $n_i$ has the minimum non-miner score so it will replace an existing outgoing edge of $m$ to another non-miner if FIFO (or with some probability if random), since $m<d$. Thus, in a stable topology, all non-miners have an equal score of 1. 
            
            Non-miners who have all $d-1$ out-going edges to all miners have that those edges are stable since the miners are stable and have the lowest score. Since miners began the protocol with $d-1$ out-going edges, each miner must have maintained at least $d-1$ edges since you only drop an edge if you have $d$ edges. There are at least $d-1-m/2$ non-miners with incoming edges from a miner. If $m<d-1$, then non-miners who connect to at least 2 other non-miners (unless they are already connected to all non-miners which is not many since $n \gg d$). Thus take one of these non-miners, in the next round they connect to a non-miner that they were not already connected to, and since all non-miners have the same score, they replace one of their current non-miner connections with this new connection if FIFO (or with some probability, if random). Now consider $m=d-1$, some non-miners have all $d-1$ stable outgoing connections to all miners, thus at any point only one non-miner can connect to all non-miners (i.e. using the $d$-th edge of most non-miners). There is at least one miner with at most $\frac{m-1}{2}$ out-going connections, thus at least $d-1-\frac{m-1}{2}=d/2$ non-miners have at least one incoming connection from a miner. One of these non-miners is not connected to some other miner $m_i$, and in the next round could replace one of their existing connections to a non-miner with a connection to $m_i$ (with some probability). Thus there is always an unstable edge from the first $d-1$ out-going edges for some node.  

            \noindent(Case 2: $d\leq m < 2d$, FIFO or Random:) By Proposition~\ref{prop:only-stable-clique}, we know that for a topology to be stable, miners must connect in a clique which uses $\frac{m(m-1)}{2}$ edges. This leaves $m(d-1)-\frac{m(m-1)}{2}$ edges from the miners going out to non-miners. Assume that there exists at least one non-miner $n_i$ who is not connected to all miners. Node $n_i$ must have all outgoing edges to miners, otherwise, there is a miner it could connect to which would replace the edge since miners have the minimum score. Since $n_i$ is not connected to all miners, in the next round it could connect to some miner and this miner would (or could with some probability) replace an existing edge since they all have the same score. Now imagine all non-miners connect to all miners, it must be that all non-miners have all outgoing edges to all miners, otherwise, an outgoing edge to a non-miner can always be replaced (since there are not enough edges for non-miners to form a clique). Thus the only stable topology would be the unique case when the number of non-miners $n'$ is exactly such so that all non-miners connect to $d-1$ miners and have $m-(d-1)$ incoming connections from miners, thus
                $$n'(m-d+1)=m(d-1)-\frac{m(m-1)}{2}=\frac{m}{2}(2d-m-1)$$
            we get
                $$n = m+ n'=\frac{m}{2}\frac{(2m-2d+2)+(2d-m-1)}{(m-d+1)}=\frac{m}{2}\frac{m+1}{m-d+1}< \frac{dm}{m-d+1}<md<2d^2$$
            
            \noindent(Case 3: $m > 2d-1$, FIFO or Random:) By Proposition~\ref{prop:imposs-ran-stable}, there is no miner stable topology, so there is no stable topology.
            
    \end{proof}

    All examples given in the affirmative proofs of Table~\ref{tab:stab-results} rely on specific topologies involving a highly connected sub-network of miners, and generally low-diameter networks. Next, we show that such properties are in fact necessary for any stable network, and briefly explore the impact of inbound connection caps.

     \subsubsection{Properties of stable topologies}
    \label{app:stable-properties}
        We can extend the above analysis to derive the following central properties that \textit{any} miner-stable or network-stable topologies must have. Principally, these state that stable networks must necessarily have a highly-connected miner component.
        
        \begin{proposition}
                If $m < 2d$, and ties are broken by FIFO or randomly, the only miner-stable topology is a clique.
            \label{prop:only-stable-clique}
        \end{proposition}
        \begin{proof}
                The proof follows similarly to that of Proposition~\ref{prop:imposs-ran-stable}. Assume there is some non-clique miner-stable topology. Thus there is some miner $m_i$ who is not connected to all miners. There is some probability that in the next round, all miners connect to this miner with at least one miner, $m_j$ having a new out-going edge to $m_i$. Miner $m_i$ now has score $\frac{m-1}{m}$, thus $m_j$ will not drop the edge to $m_i$ since it is either smaller than some existing edge or ties with all edges and by FIFO (or random) \tiebreak replaces some edge (has a probability of replacing some edge). 
        \end{proof}
    
        \begin{proposition}
                For all $m,d$ and a global order (or LIFO) \tiebreak, in any miner-stable topology, all miners must connect to the first (or some) $d-1$ miners. 
            \label{prop:first-d-1-stable}
        \end{proposition}
        \begin{proof}
                (Case 1: global ordering) Assume there is some miner-stable topology where there is a miner $m_i$ who does not connect to some miner $m_j$ where $m_j$ is in the top $d-1$ miners. In any round, there is a chance that all miners connect to miner $m_j$ with miner $m_i$ forming an outgoing edge to $m_j$. Thus in this round $m_j$ has the minimum score and won't be dropped by $m_i$ since it either has a lower score than some outgoing edge of $m_i$ or is ranked higher.
    
                (Case 2: LIFO) Assume that there are less than $d-1$ miners who all miners connect to in a miner stable topology. There is some miner $m_i$ with score $>
                \frac{m-1}{m}$ which some miner, $m_j$, is not connected to. There is a chance that in the next round all miners who are not yet connected to $m_i$, connected to it, including $m_j$ with an outgoing edge to $m_i$. In this round, $m_i$ has the minimum score of $\frac{m-1}{m}$, and since less than $d-1$ miners have such a score, $m_i$ will now be in the top $d-1$ edges of $m_j$.
        \end{proof}
    
        \begin{proposition}
                Any miner-stable topology (regardless of $m,d$) must have miner diameter at most $2$ and will lead to a network diameter of at most $3$.
            \label{prop:net-degree}
        \end{proposition}
        \begin{proof}
                This follows directly from requirements of stability. Assume there is some miner-stable topology where no miner connects to all miners. There is always some possibility that in a round all miners connect to some miner $m_i$ with some miner $m_j$ having a new outgoing edge to $m_i$. As has been argued, this new edge contradicts that the initial topology was miner-stable. Thus any stable topology has at least one miner who connects to all other miners. Since the lowest-scoring node is a miner, eventually all non-miners will connect to at least one miner and any future state has all non-miners connecting to at least one miner. Note there is no guarantee of stability of these edges, but we get that the distance between any two non-miners is at most 3.
        \end{proof}

        Recall from Lemma~\ref{lem:stable-nash} that the Nash Equilibrium we present for the idealized game also has a miner diameter of 2 and network diameter of 3.

    \subsubsection{Inbound connection cap} We briefly explore the impact of inbound connection caps on the stability properties we explored above. Without loss of generality, we assume that the inbound cap $d_{in}>d$, otherwise nodes would not have the inbound capacity to form new connections. 

    \begin{lemma}
        For  $ m< 2d$ and any \tiebreak, there exists a miner-stable topology. For tie-breaking rules LIFO and global ordering, miner-stable topologies exist for $m\leq max(2d-1,d_{in}+d/2)$. 
        \label{prop:capped_same_stab}
    \end{lemma}
    \begin{proof}
        For $m < 2d$, there are enough edges between miners for cliques to form, thus they are still stable under Proposition~\ref{prop:miner-clique}.
        For $m\geq 2d$, consider a topology where the first $d-1$ miners form a clique, call them set $M$, then the rest of the miners have their outgoing edges to all miners in $M$. Thus all miners in $M$ have the minimum score, and won't be dropped by any miner; all miners not in $M$ have their outbound capacity filled by these stable edges so they will not connect to each other. We now calculate the maximum miners not in $M$. The maximum in-bound capacity for all miners in $M$ is for each miner to have equal incoming/outgoing edges $=\frac{d-2}{2}$. Thus they each have $d_{in}-\frac{d-2}{2}$ capacity left for miners not in $M$, this happens for $m\leq d-1+d_{in}-\frac{d-2}{2}=d_{in}+d/2$.        
    \end{proof}

    It remains an open problem if there exists a miner-stable topology for other regimes.

\subsection{Reachable Stability}

    In the previous section, we characterized the exact conditions for when miner-stable and network-stable topologies exist for the unbounded \greedy protocol under different \tiebreaks. We next consider \textit{if} these stable topologies are necessarily reachable. We show that \textit{for} $m\leq d$, \textit{all \tiebreaks will lead to miner-stable topologies}, in particular a miner-clique. Additionally, for $m<d$, we show that\textit{ if there exists a stable network, then the network will stabilize}. For larger values of $m$, our proofs of stability in the previous subsection lead us to believe that the stable networks that do exist are hard to reach, we leave this for future work. 
    
    We consider the state transition model of the \greedy game. We define a state of the protocol as the graph $G$ of the network at some round $r$, and the set of reachable states as the states possible after the execution of one round of the protocol. 

    \begin{proposition}[Eventual Miner-Stable Clique]
        If the number of miners $m\leq d$, the network will reach a miner-stable topology.
        \label{prop:eventual_clique}
    \end{proposition}
    \begin{proof}
        To prove this we show that some miner-stable topology is reachable from any state (network topology and round). To do this we present the following \textit{clique process} which defines an exact path from \textit{any} state to some miner clique which we've shown in the previous section is the only miner-stable topology for $m \leq d$. For convenience, we define some arbitrary order over all nodes such that the miners are the first $m$ nodes.

        \noindent \textbf{Def. \textit{clique-process}:} \\
        \indent For each round $r$, take $\min_i$ s.t. score$(m_i) > \frac{m-1}{m}$: \\
        \indent\indent 1. $\forall j\leq m$, if $e_{i,j}=0$, then $m_j$ chooses $m_i$ as their new random outgoing edge. \\
        \indent\indent 2. $\forall j\leq i$ score$(m_j) = \frac{m-1}{m}$ \\
        \indent\indent 3. Note: $\forall j\leq i$ no edges \textit{to} $m_j$ are dropped
        \\
        For the remainder of the proof, we re-define an edge as stable in some state, if the value of the edge does not change for all future states on the path defined by the clique-process. Equally, a miner is miner-stable if all of its edges to other miners are stable in the set of future states defined by the clique-process path. Note that once all miners are miner-stable on the clique-process path, then the network is miner-stable for any path as this only happens when a clique is formed by the miners.

        We now prove the following statement: Starting from any state, the path defined by the clique-process will result in the first $j$ miners being miner-stable, for $j\leq m$.
        
        \textbf{Base case: (Miner $m_1$ will miner-stabilize)} Assume that $m_1$ never stabilizes. This means $\exists j\leq m$ s.t. $e_{1,j}$ does not stabilize. By definition of the clique-process, $i\geq 1$ for all rounds, thus at step (2.) of each round score($m_1$)=$\frac{m-1}{m}$, thus $e_{1,j}=1$. For this edge to not be stable, $\exists r' \geq r$ where $e_{1,j}=0$ at step (3.) which only happens if this edge was an outgoing edge from $m_1$, and $m_1$ drops the edge at round $r'$. Then, at round $r'+1$ we have $i=1$, so miner $m_j$ will form a \textit{stable} outgoing edge to $m_1$ (note step 3 of the clique-process). 

        \textbf{Inductive step:} Assume that $\forall j'< j\leq m$ miner $m_{j'}$ is miner stable at some round $r$. It must be that $\forall j'< j$ score$(m_{j'})=\frac{m-1}{m}$. Note that this means for each round $r'\geq r$, score$(m_{j})=\frac{m-1}{m}$ at step (2.), even if $m_j$ is not stable, as it must be that $i\geq j$. The proof that $m_j$ stabilizes follows similarly to the base case.
        
        Assume $m_j$ never stabilizes, this means $\exists h\leq m$ s.t. $e_{j,h}$ does not stabilize. By definition of the clique-process, $i\geq j$ for all rounds, thus at step (2.) of each round score($m_j$)=$\frac{m-1}{m}$, thus $e_{j,h}=1$. For this edge to not be stable, $\exists r' \geq r$ where $e_{j,h}=0$ at step (3.) which only happens if this edge was an outgoing edge from $m_j$, and $m_j$ drops the edge at round $r'$. If this happens, at round $r'+1$ we have that $i=j$, so miner $m_h$ will form a \textit{stable} outgoing edge to $m_j$ (note step 3 of the clique-process). 
    \end{proof}

    We now extend our results to include network stability guarantees. Recall that only \tiebreak rules LIFO and Global-ordering have stable networks and only for $m<2d$. Note that the following proposition is for $m$ strictly less than $d$ as the stable networks we used in our proof of Lemma~\ref{lem:stable-top} for $d\leq m < 2d$ were quite contrived and may not be reachable.

    \begin{proposition}[Eventual Network Stability]
            If the number of miners $m < d$, given \tiebreak of Global-ordering or LIFO, the network will arrive at a stable topology.
            \label{prop:stable-inevitable}
        \end{proposition}
     \begin{proof}
        From Proposition~\ref{prop:eventual_clique} we get that the miners will eventually form a miner-stable topology which is a miner clique. Given this, in the remainder of the proof, we first argue that all edges involving miners will stabilize, and lastly that all edges between non-miners will stabilize.

        \textbf{1. $\forall i\leq m, e_{i,j}$ stabilizes for all $j$.} To ease our analysis, we present a process which defines a path from any miner-stable topology to all edges involving miners being stable. For LIFO, we define some arbitrary order over all nodes, note it is important non-miners choose their edges before the miners do. In both cases, if there are no such nodes that satisfy the process, the node chooses a random new edge.

        \noindent \textbf{Def. \textit{miner-stability-process}:} \\
        \indent For each round $r$:\\
        \indent\indent - Each non-miner chooses as their out-edge, the first miner it is not yet connected to.
        \indent\indent - Each miner chooses an their out-edge the first non-miner it is not connected to.

        We prove that $\forall i\leq m, \exists r$ s.t. $e_{i,j}$ is stable $\forall j$, via a strong induction proof.

        \textbf{Base case: (All of miner $m_1$'s edges will stabilize)} In the first round of the process, all non-miners $n_i$ not connected to  $m_1$, connect to it and thus $e_{1,i}=1$ with $n_i$ never dropping $m_1$ (miners have the smallest score and $d>m$). Assume there is some edge of $e_{1,j}$ which never stabilizes. This edge must be from $m_1$ to $n_j$ (otherwise $n_j$ would have connected to $m_1$ in some round), and so has value 1. For it not to be stable, at some point it must be dropped by $m_1$ at step 3. of some round $r$. If so, in round $r+1$ node $n_j$ would then make a stable edge to $m_1$ $\bot$.

        \textbf{Inductive step:} Assume that $\forall j' < j\leq m$ miner $m_{j'}$ is  stable at some round $r$. It must be that miner $m_j$ is connected to all nodes for $r'\geq r$, and note that if a non-miner connects to $m_j$, they never drop the connection as, again, miners have the smallest score and $d>m$. The proof follows similarly to the base case:

        Assume there is some edge of $e_{j,i}$ which never stabilizes. This edge must be from $m_j$ to $n_i$ (otherwise $n_i$ would have connected to $m_j$ in some round), and so has value 1. For it not to be stable, at some point it must be dropped by $m_j$ at step 3. of some round $r$. If so, in round $r+1$ node $n_i$ would then make a stable edge to $m_j$ $\bot$.

        We now have that all miner edges stabilize, what's left to show is that the non-miner edges will stabilize. 
        
        \textbf{2. $\forall i,j e_{i,j}$ stabilizes.} We've shown that all miners connect in a stable clique (have score of $\frac{m-1}{m}$, and all non-miners have stable connections to all miners (have score of $1$). Note for LIFO \tiebreak, all edges between non-miners are stable as they all have equal lowest-score which will never change. We just need to show that the non-miner edges for Global-Ordering \tiebreak stabilize. Again we define a stabilizing process:

        \noindent \textbf{Def. \textit{non-miner-stability-process}:} \\
        \indent For each round $r$:\\
        \indent\indent - Each non-miner who does not have $d-1$ stable edges, picks the smallest-indexed \\ \indent\indent ~~non-miner they are not connected to \\
        \indent\indent - All other non-miners then pick a random new node

        Again use a strong induction argument to prove all non-miners stabilize.

        \textbf{Base case: the smallest indexed non-miner $n_{m+1}$ stabilizes} Notice that in the first round of the process, all non-miners connect to $n_{m+1}$. Since it is smallest indexed non-miner, all incoming edges to it will not be replaced. Since it is connected to all nodes, and all incoming edges are stable, there are no other nodes, and will never be any other nodes, it could connect to as their $d$-th edge, so no outgoing edges will be dropped.

        \textbf{Inductive step:} Assume that $\forall  j' < j$ miner $n_{j'}$ is stable at some round $r$. Assume $n_j$ will never stabilize, this means there is some round $r'\geq r$ where $n_j$ will not gain any more stable edges. Assume in this round, some non-stable node $n_i$ connects to $n_j$. This edge will be stable as this node must already be connected to all $n_{j'}$ for $j'<j$, otherwise it wouldn't have chosen $j$, and all those connections are stable by the assumption, so there is no future edge that could win the tiebreak. If no new nodes pick $n_j$ as their outgoing edge, this means $n_j$ picked some node $n_i$ it was not connected to (otherwise it would have been stable) as its new edge. If $n_i$ has a smaller index than an existing edge, then it replaces that edge and stabilizes (as all incoming connections are stable). If $n_i$ does not replace an existing edge, this means all outgoing edges were stable. In every case we get a contradiction $\bot$.

     \end{proof}

    \begin{corollary}
        \label{corr:eventual-stable}
        For $m<d$, and any \tiebreak, if there exists a miner-stable or network-stable topology, then the network will arrive at some such topology.
    \end{corollary}

     \junk{
        \subsection{Markov Model}
        To ease our analysis, we define the following Markov Model to capture all possible states of the protocol.
    
        \begin{definition}[States of the Markov model for \greedy]
        We define a state $S_\alpha$ for a graph $G_\alpha$ with edges $e_{i,j}\in E_\alpha$ which encode the information needed for the \tiebreak. Each round of \greedy is a state transition. 
    
        We define the transition $E(G^\beta_{S_\alpha}):S_\alpha \Longrightarrow \{S_\beta\}$  such that:
        \begin{enumerate}
            \item After step one of \greedy, we end up with the graph $G^\beta_{S_\alpha}$
            \item $\{S_\beta\}$ if the set all of states represented by each graph $G_{S_\beta}$ that results from step three of \greedy. If the \tiebreak is deterministic, this set contains only one state. 
        \end{enumerate}
        \end{definition}
    
        We can define the proposition $P(G)$ as true if the graph $G$ satisfies the property $P$, and $P(\{G\})$ is true if all graphs in $\{G\}$ satisfy the property.

        \begin{definition} [Terminal States]
            We define a state $S_\alpha$ as \textit{terminal} if, for all transitions $E(G^\beta_{S_\alpha})$, $~G_{S_\alpha} = G_{S_\beta}$. We define \minerterminal states as the set of states where the property of miner stability holds (all edges between miners are stable), and \networkterminal as the set of states where all edges between any two nodes are stable.
        \end{definition}   
     }

\section{Simulations}
\label{sec:simulations}
    In the previous section, we prove the conditions for stability in a network implementing \greedy and some conditions for when a network running \greedy will stabilize. In this section, we simulate \greedy to explore how the networks evolve, possible properties that stabilize for different network sizes and in-bound connection capacities, and the impact of having more exploratory edges in the protocol. 

    As in the theory, we assume each node can uniformly pick a random new peer from the set of all nodes in the network, and knows the score of their peers at step (3) of each round. In simulations, at step (1) of the protocol, each node adds a new edge to a peer they are not yet connected to. The order of the peer selection is randomized in each round, and with capped $d_{in}$, if the new peer they try to connect to has full in-edge capacity, then the node tries another random node until they succeed, or run out of peers (\eg if all nodes that are not currently peers have full in-connections)\footnote{We emphasize that all other steps of each round are done in parallel for all nodes. We explored simulations where nodes completed each full round one at a time and found no differences in the simulation results.}. In the previous section, we prove that of the \tiebreaks we consider, only LIFO and a global ordering have the property that for all $m,n$ there exists a miner-stable topology. Since a global ordering over all nodes is not intuitive in a decentralized network, we  use LIFO as the \tiebreak in our simulations.  We run each of the simulations below 20 times and plot the averages of the runs. 
    
    \noindent \textbf{Main results of the simulations:}
    \begin{NoIndentEnumerate} 
        \item 
        Our simulations suggest that the effect of the $d_{in}$ cap is smooth (as the cap is raised, the behavior shifts smoothly to the no cap model).
        
        \item Though theoretical results show the existence of stable topologies, in particular for $m \leq d$ that cliques between miners are stable and inevitable, even for low $n$ (e.g., 400 and 900), arriving at such a topology is hard/improbable and our simulations do not reach a miner clique in the first 256 rounds.

        \item An extension of \greedy 
        where nodes drop the $k$-worst outgoing edges each round suggests the additional randomness can lead to the network converging quicker to the same values, but that this advantage is maximized at low values of $k$ (e.g., 3).
    \end{NoIndentEnumerate}
    \textit{Our analysis indicates that, consistent with the theoretical results, a hierarchical structure emerges within these networks with respect to node scores, degree, and centrality, which are highly correlated. In general, these networks converge to overall lower properties, but often with some nodes (primarily the miners) being at a clear advantage and skewing the average.}\\
    
    \noindent \textit{Starting topology.} For most of our simulations, we begin with a random graph where each node chooses $d$ initial random peers it is not yet connected to
    (see Figure~\ref{fig:graphs}(a).). 
    As we discuss below, the \greedy protocol tends to converge to a structure like Figure~\ref{fig:graphs}(b) where few nodes in the middle connect to most others on the periphery. 
    We also run simulations starting with a small-world graph (using the Watts-Strogratz \cite{watts1998collective} model with probability 0.5) and with a scale-free graph (using the Barabasi-Albert \cite{albert2002statistical} model with an initial connected component of size 20). For full analysis, see Appendix~\ref{app:alt-graph}. We observe that the initial graph does have a small impact on the properties the \textit{capped} simulations are converging to (\eg scale-free network converges to higher average distance to miners).

    \begin{figure}
            \centering
            \begin{subfigure}{.32\textwidth}
                \centering
                \includegraphics[width=\linewidth]{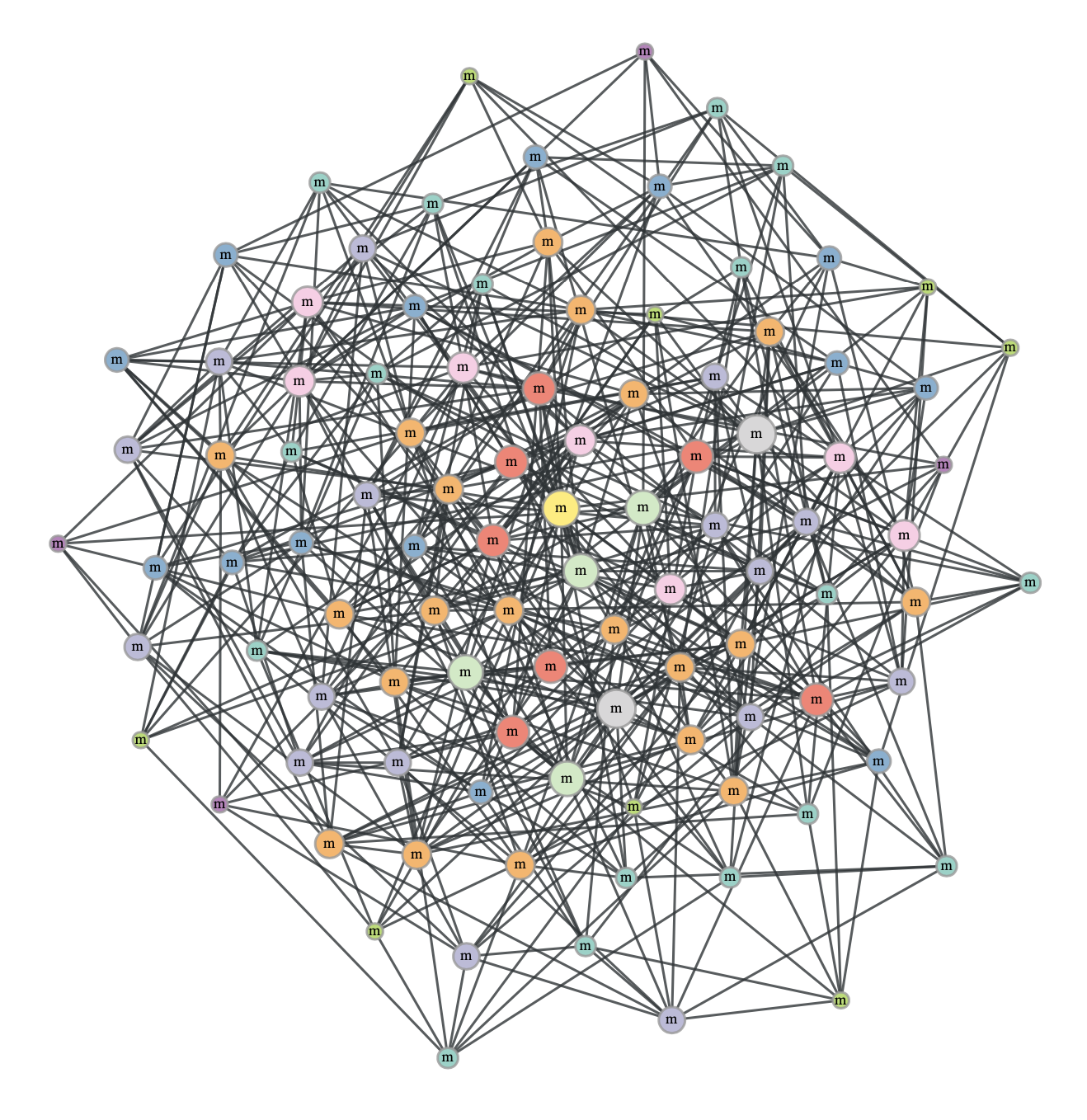}
                \captionsetup{justification=centering}
                \caption{}
            \end{subfigure}~~~~~~~~~~~~~~~~~
            \begin{subfigure}{.32\textwidth}
                \centering
                \includegraphics[width=\linewidth]{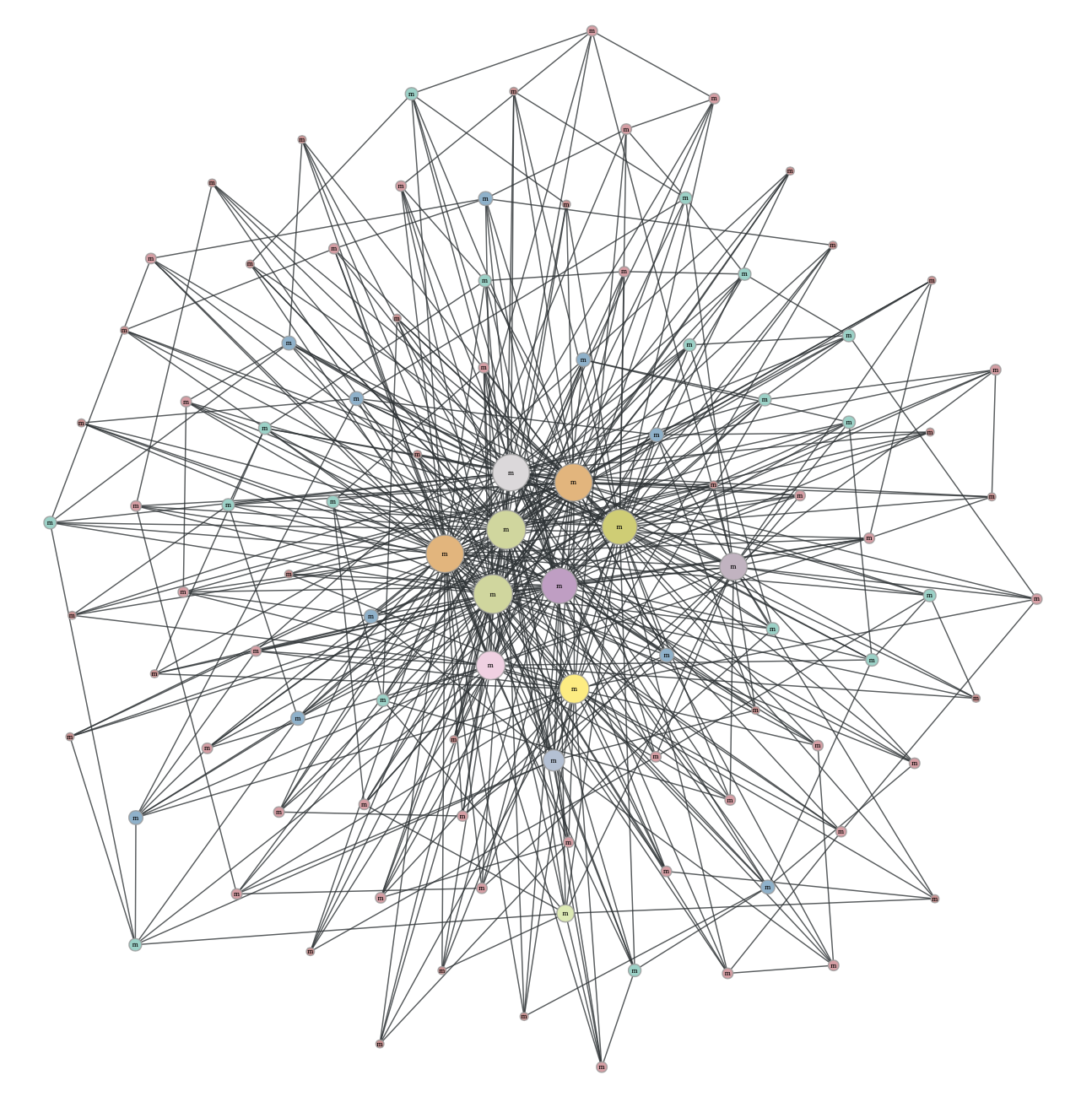}
                \captionsetup{justification=centering}
                \caption{}
            \end{subfigure}
            \caption{
            (a) the random graph created from all nodes choosing $d$ random nodes to connect to. (b) the graph we seem to converge to from running \greedy with all nodes being homogeneous miners, size scaled by degree. This structure appears after $\sim$20 rounds. 
            } 
            \label{fig:graphs}
    \end{figure}

    \subsection{Effect of In-Degree Caps}
        To complement our limited theoretical analysis with bounded $d_{in}$, we explore the impact the inbound cap has on networks of all miners, and networks with a subset of 10 nodes being miners. Figure~\ref{fig:mult-avg-d} shows the average distance over time for different inbound caps. We see that if the network is all miners, we need a substantial inbound cap for the average distance between all nodes to be better than the random graph within the span of our simulations. We see, however, in Figure~\ref{fig:mult-avg-d-pdf} that this lowered average distance is skewed by a few nodes being better off than the majority of the network. When only some of the nodes are miners,  the network has the potential to converge to one with a better average distance to the miners; higher the inbound cap,  lower the average distance to miners. 
        Figure~\ref{fig:mult-diam} supports this, showing the network diameter approaching 2. Additionally, the diameter between just the miners is approaching 1 meaning that even with bounded capacity, the miners are forming a clique. \textit{Our main observation here is that the value of the bound has a smooth effect on network properties.}


        \begin{figure}
            \centering
            \begin{subfigure}{.4\textwidth}
                \centering
                \includegraphics[width=\linewidth]{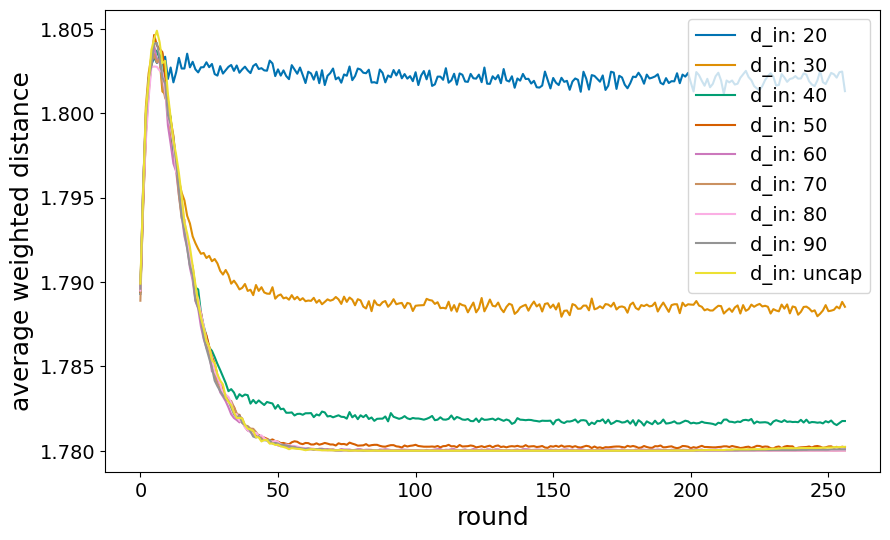}
                \caption{All Miners}
            \end{subfigure} ~~~
            \begin{subfigure}{.4\textwidth}
                \centering
                \includegraphics[width=\linewidth]{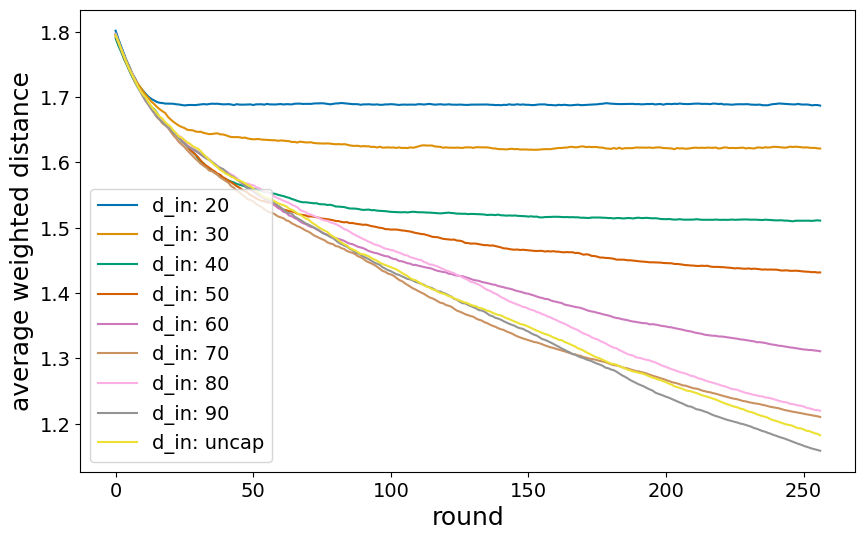}
                \caption{Subset Miners}
            \end{subfigure}
            \caption{Average distance to miners per round for an $n=100$ node network with all miners and $m=10$ miners, for different in-degree caps. }
            \label{fig:mult-avg-d}
        \end{figure} 

        \begin{figure}
            \centering
            \begin{subfigure}{.4\textwidth}
                \centering
                \includegraphics[width=\linewidth]{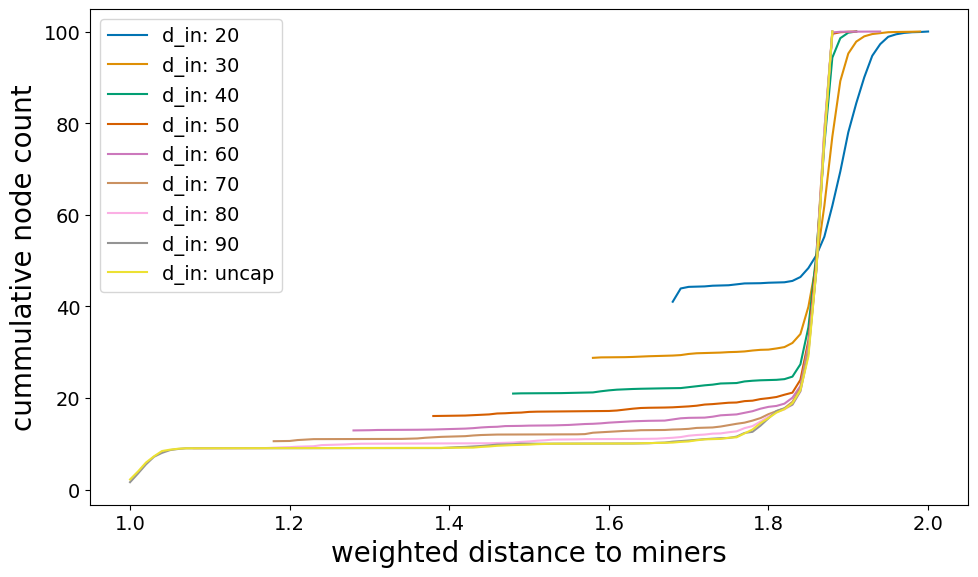}
                \caption{All Miners}
            \end{subfigure}~~~
            \begin{subfigure}{.4\textwidth}
                \centering
                \includegraphics[width=\linewidth]{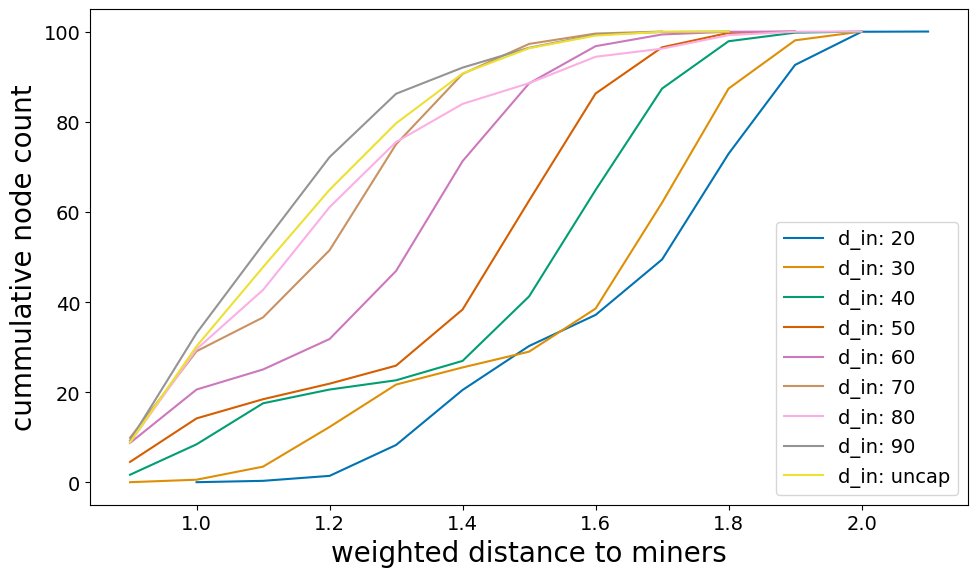}
                \caption{Subset Miners}
            \end{subfigure}
            \caption{Distribution of distance to miners(node scores) at round 256 for an $n=100$ node network with all miners and $m=10$ miners, for different in-degree caps.}
            \label{fig:mult-avg-d-pdf}
        \end{figure} 

        \begin{figure}[H]
            \centering
            \begin{subfigure}{.75\textwidth}
                \centering
                \includegraphics[width=\linewidth]{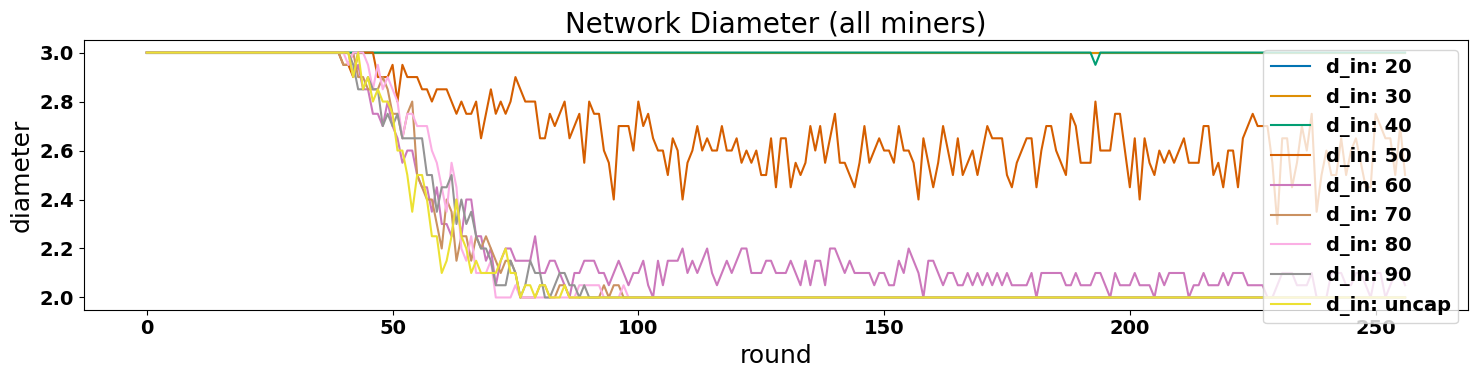}
            \end{subfigure}\\
            \begin{subfigure}{.75\textwidth}
                \centering
                \includegraphics[width=\linewidth]{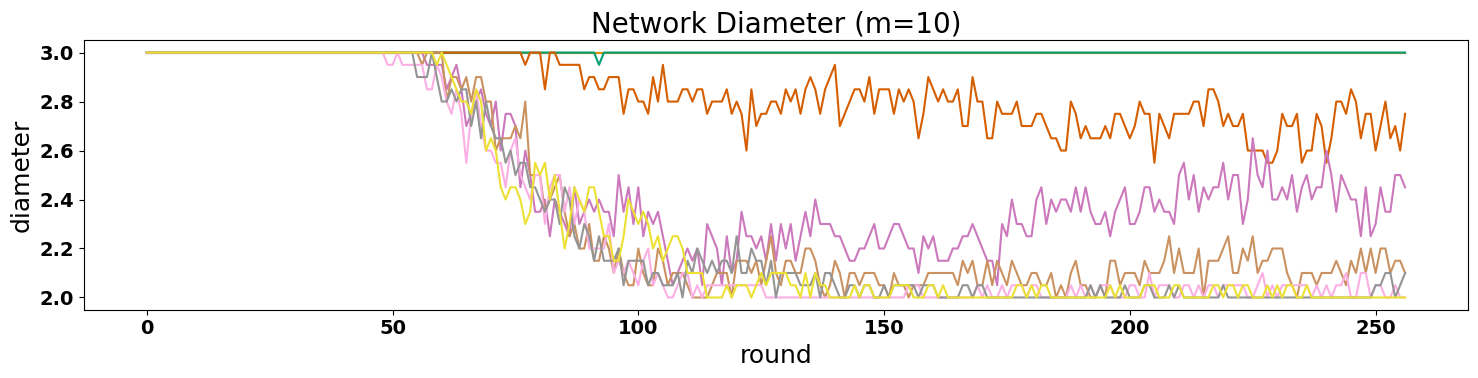}
            \end{subfigure}\\
            \begin{subfigure}{.75\textwidth}
                \centering
                \includegraphics[width=\linewidth]{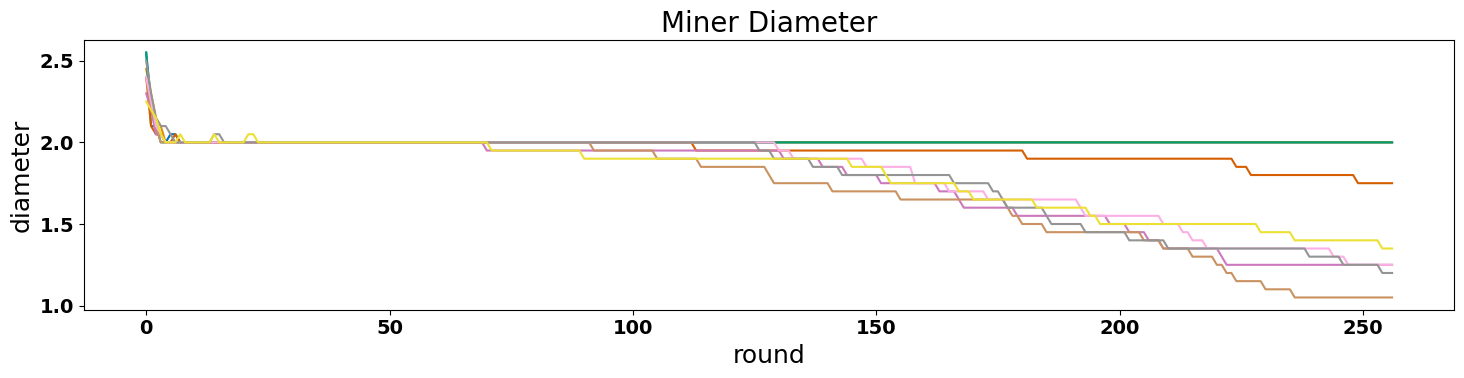}
            \end{subfigure}
            \caption{Diameter of a 100 node network with all miners and a subset of 10 miners, and the diameter between the 10 miners. }
            \label{fig:mult-diam}
    \end{figure}

    \subsection{Network Size}
    We also explore the impact of network size on our simulations when a subset of the network is miners. Again, we consider $m=10$ and $d=10$, with $d_{in}=20$ or uncapped. Our theoretical results show that for $m\leq d$, the miners connecting in a clique is a miner-stable topology, regardless of what the rest of the network is doing (for uncapped $d_{in}$, this is the only miner-stable topology). We run simulations for $n=100, 400$ and $900$ and plot the  network/miner diameter and eccentricities in Figure~\ref{fig:10-diam-eccen}.
    %
     %
     \textit{Our main take away is that, though miners are becoming better connected over time, with more non-miners, it is harder for the miners to form a clique.} This is likely due a decreasing probability of another miner being randomly selected by a miner as a neighbor as the number of nodes increases. The non-miners also maintain connections to the miners (seen in the eccentricity of the miners in Figure~\ref{fig:10-diam-eccen}), it is thus possible the non-miners will occupy all the incoming edges of the miners before they form a clique. This seems to be the case for even $n=100$ as the capped miner diameter is stabilizing at 2.

    \begin{figure}[H]
            \centering
            \begin{subfigure}{.45\textwidth}
                \centering
                \includegraphics[width=\linewidth]{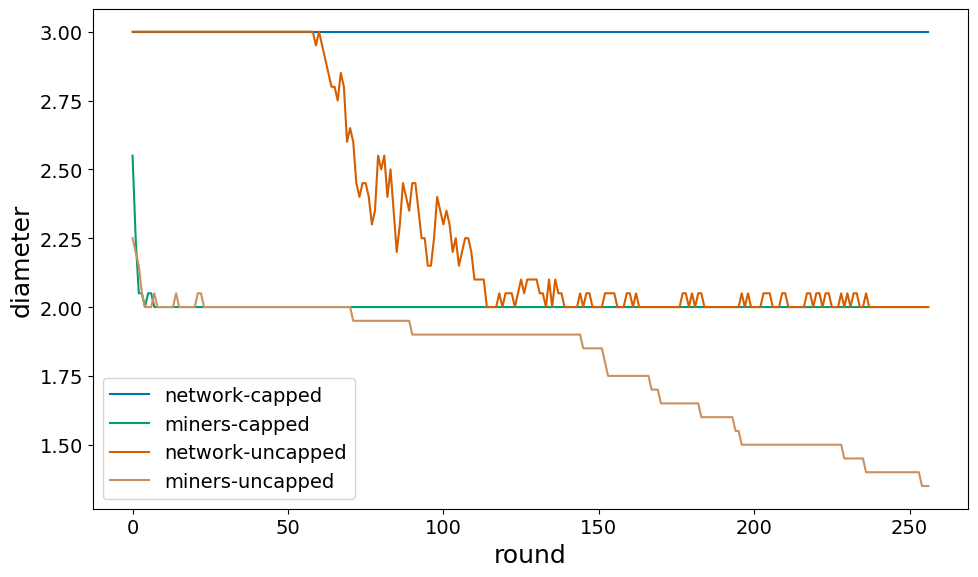}
                \caption{100 nodes diameter}
            \end{subfigure}
            \begin{subfigure}{.45\textwidth}
                \centering
                \includegraphics[width=\linewidth]{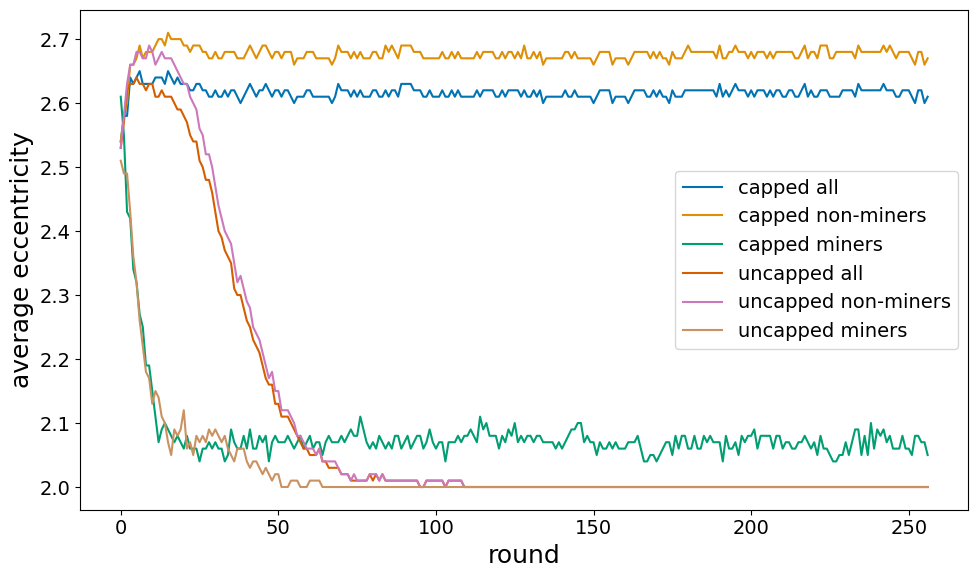}
                \caption{100 nodes eccentricity}
            \end{subfigure}\\
            \begin{subfigure}{.45\textwidth}
                \centering
                \includegraphics[width=\linewidth]{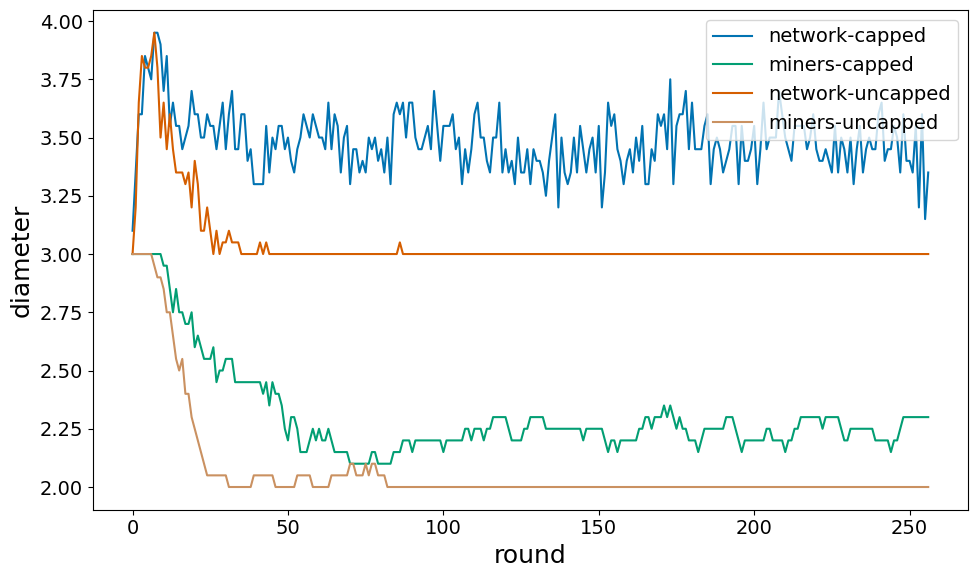}
                \caption{400 nodes diameter}
            \end{subfigure}
            \begin{subfigure}{.45\textwidth}
                \centering
                \includegraphics[width=\linewidth]{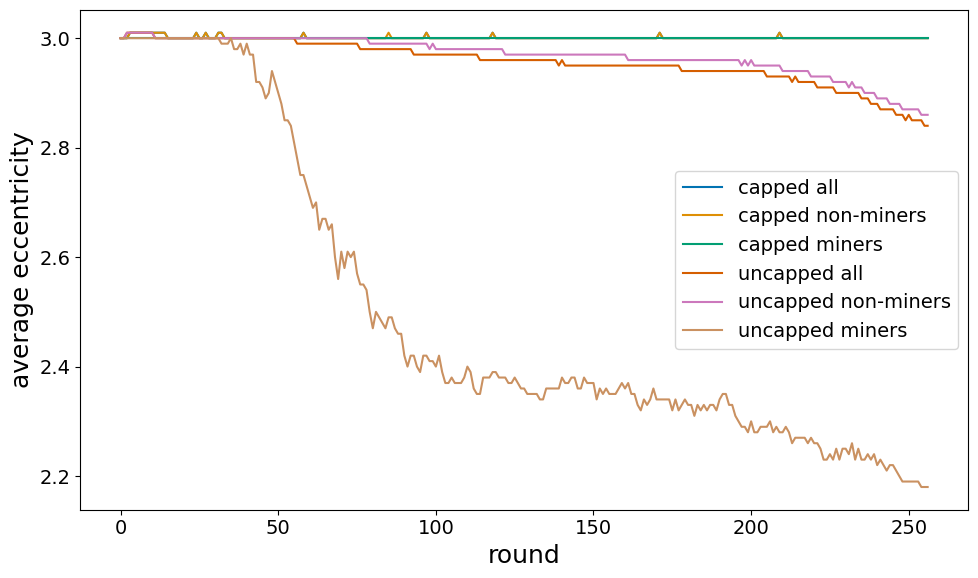}
                \caption{400 nodes eccentricity}
            \end{subfigure}\\
            \begin{subfigure}{.45\textwidth}
                \centering
                \includegraphics[width=\linewidth]{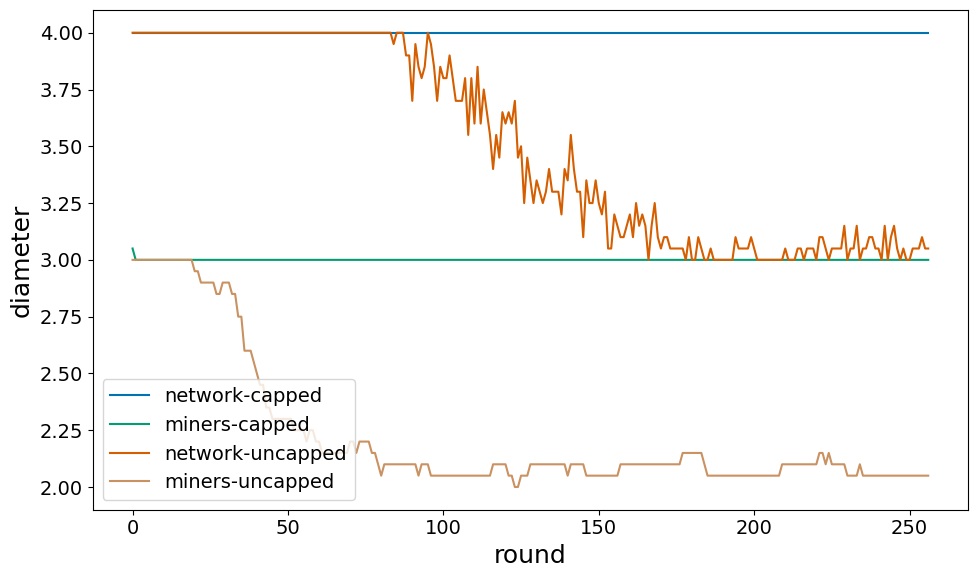}
                \caption{900 nodes diameter}
            \end{subfigure}
            \begin{subfigure}{.45\textwidth}
                \centering
                \includegraphics[width=\linewidth]{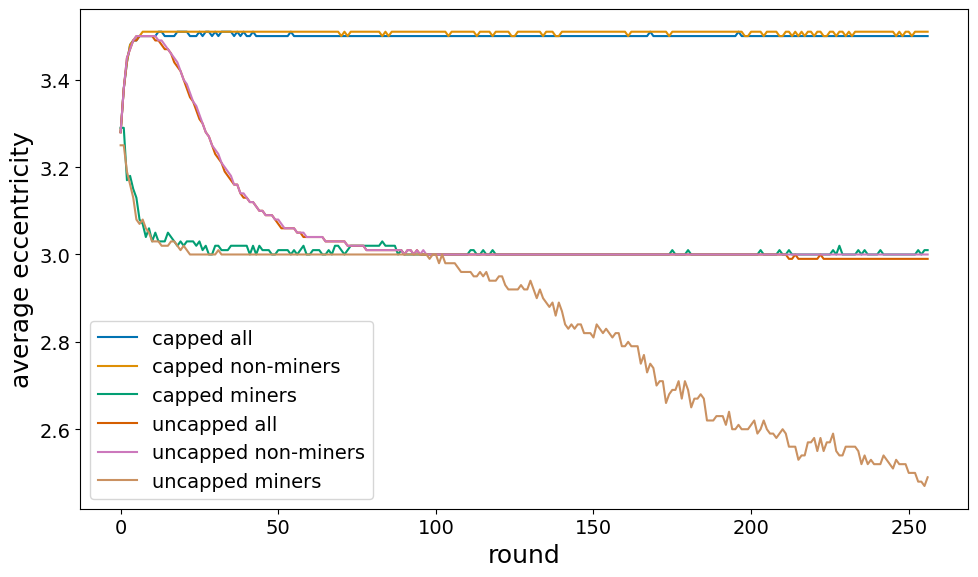}
                \caption{900 nodes eccentricity}
            \end{subfigure}
            \caption{Diameter and eccentricity over time for networks of 100, 400, and 900 nodes with 10 miners and out-degree 10. Evidence miners are mostly finding each other, it just takes some time for a clique to form.}
            \label{fig:10-diam-eccen}
    \end{figure}

    \subsection{k-worst edges} 
        So far, in both the theory and our simulations, we've considered dropping only the single worst outgoing edge. We also explore what happens when nodes drop the $k < d$ worst outgoing edges. Note that $k=d$ is essentially a random graph. Conceptually, more dropped edges leads to more randomness in the network, as more nodes have more incoming connections (see Figure~\ref{fig:din-ks}). When all nodes are miners, there is not a big difference in network behavior (Figure~\ref{fig:avg-dist-ks} a. and b. show some difference in behavior but the scale is quite small). For most cases (low enough $k$) in the subset miner simulations, miners are finding each other over time, but the additional randomness in the network appears to lead to worse average distances (Figure~\ref{fig:avg-dist-ks} d.). The main advantage that higher values of $k$ sometimes give is a faster convergence to the same values, as seen in network and miner diameters (Figure~\ref{fig:diam-ks}). \textit{Generally, dropping more outgoing edges does not change the behavior of the network much, and where it does, the optimal $k$ seems to be up to 3.}

        \begin{figure}
            \centering
            \begin{subfigure}{.4\textwidth}
                \centering
                \includegraphics[width=\linewidth]{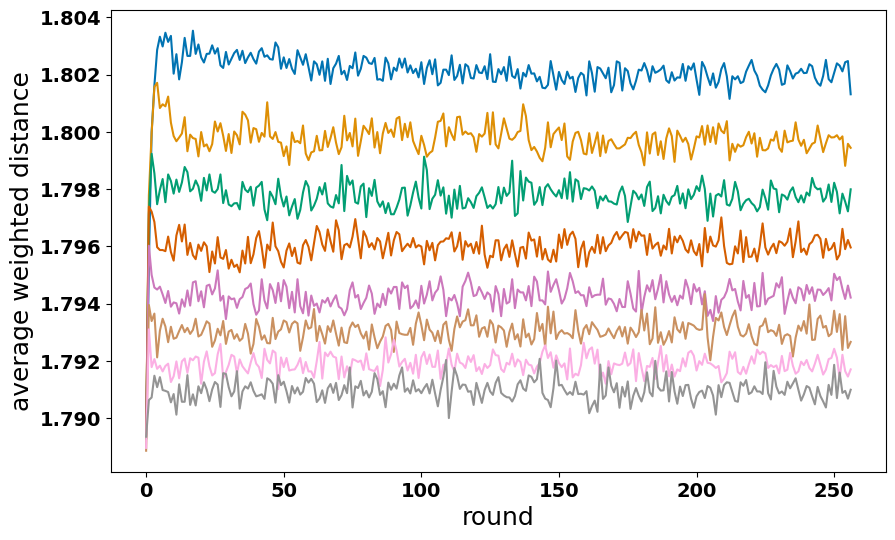}
                \caption{all miners capped}
            \end{subfigure}~~~
            \begin{subfigure}{.4\textwidth}
                \centering
                \includegraphics[width=\linewidth]{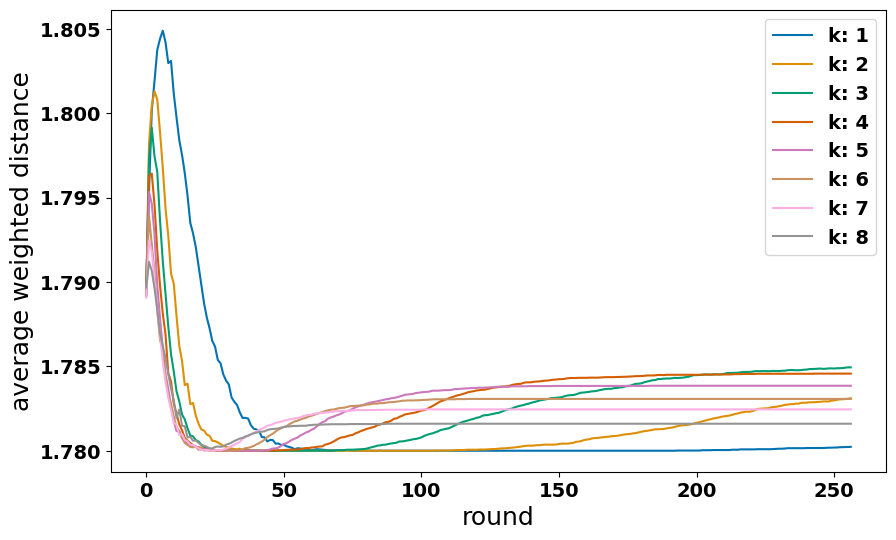}
                \caption{all miners uncapped}
            \end{subfigure}\\
            \begin{subfigure}{.4\textwidth}
                \centering
                \includegraphics[width=\linewidth]{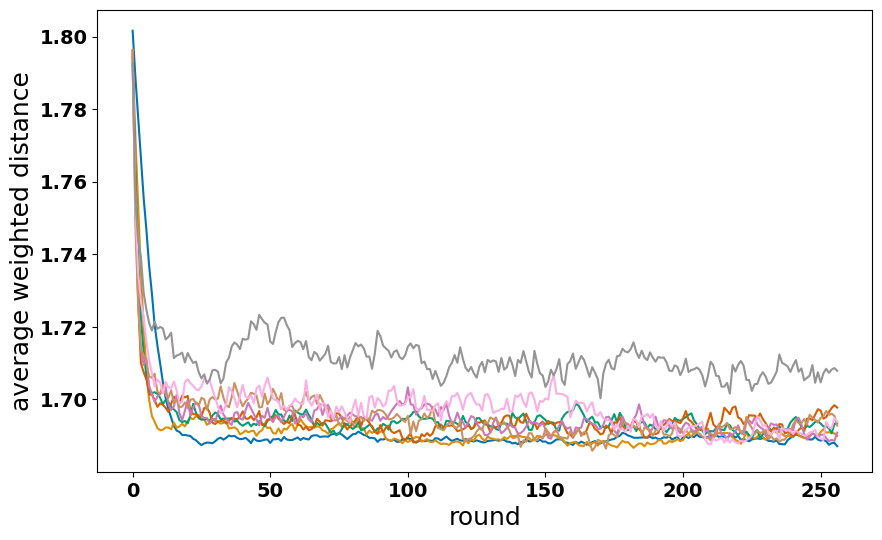}
                \caption{10 miners capped}
            \end{subfigure} ~~~
            \begin{subfigure}{.4\textwidth}
                \centering
                \includegraphics[width=\linewidth]{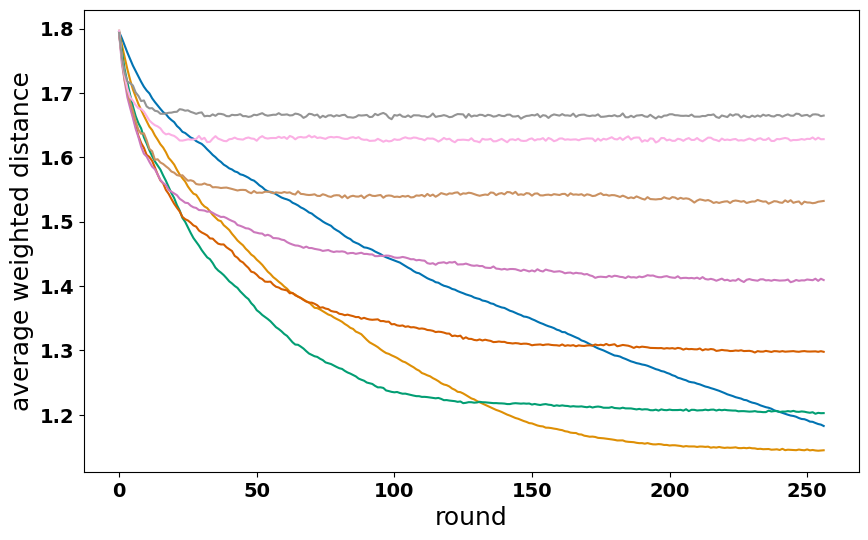}
                \caption{10 miners uncapped}
            \end{subfigure}
            \caption{Average distance to miners for different $k$-worst values, for networks of $n=100$, incoming cap of 20 or unlimited, and all miners or subset miners.}
            \label{fig:avg-dist-ks}
        \end{figure}

    \begin{figure}[H]
            \centering
            \begin{subfigure}{.75\textwidth}
                \centering
                \includegraphics[width=\linewidth]{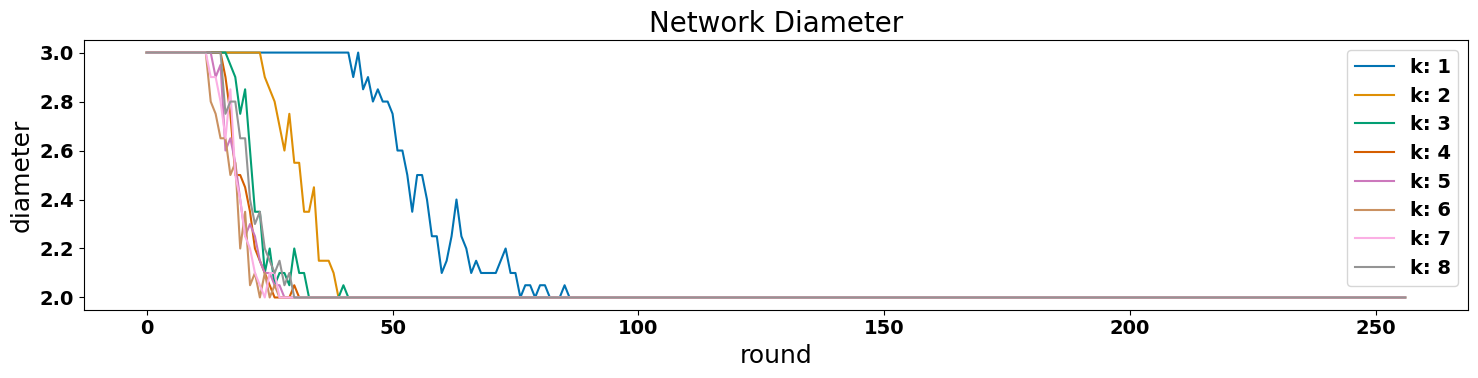}
            \end{subfigure}
            \begin{subfigure}{.75\textwidth}
                \centering
                \includegraphics[width=\linewidth]{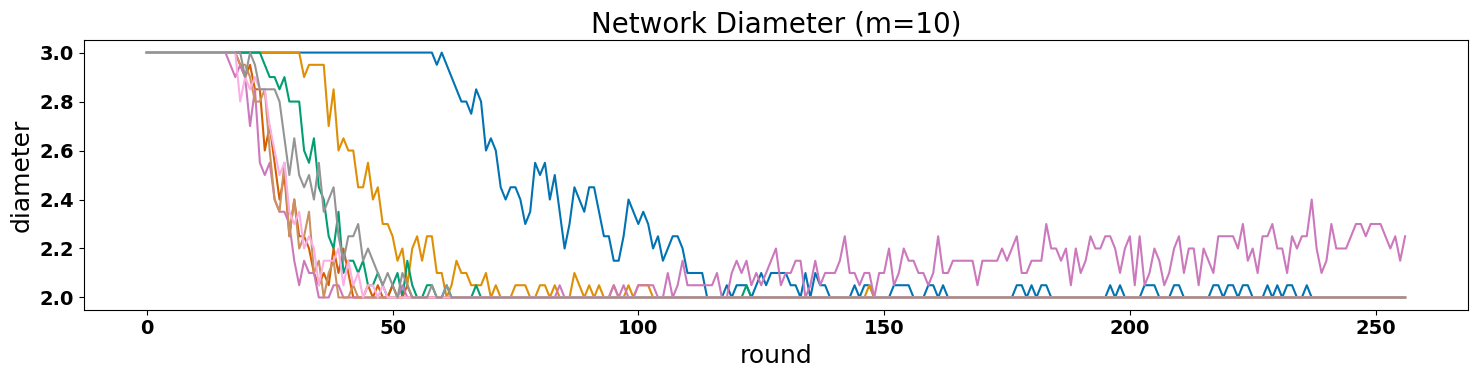}
            \end{subfigure}
            \begin{subfigure}{.75\textwidth}
                \centering
                \includegraphics[width=\linewidth]{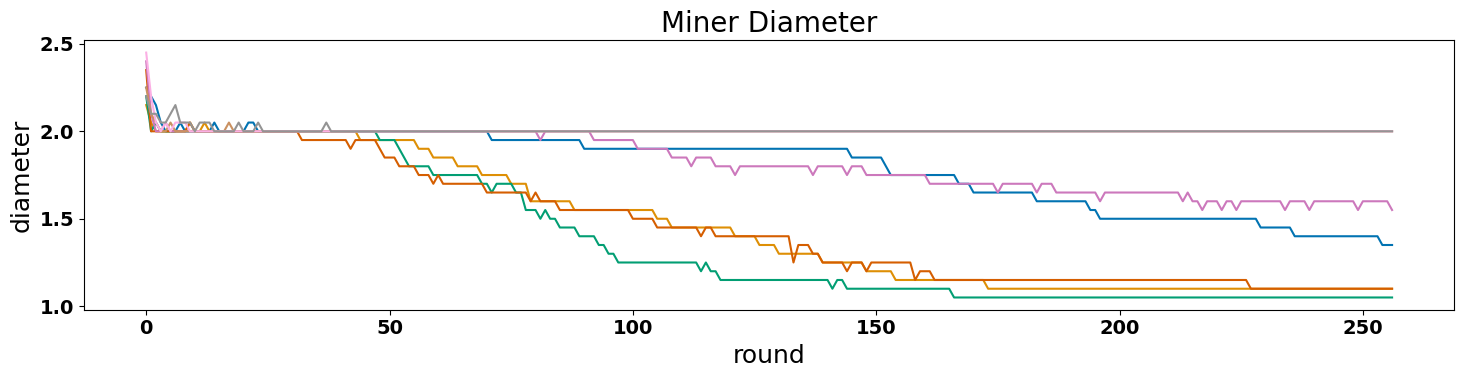}
            \end{subfigure}
            \caption{Network diameter for all miners, subset miners and miner diameter for subset miners for different $k$-worst values. All three plots are for uncapped $d_{in}$, when $d_{in}$ is capped the network diameter is constant. We see the network diameter is generally approaching 2, and that higher $k$ leads to quicker convergence. Interestingly, the miner diameter is minimal at $k=3$.}
            \label{fig:diam-ks}
    \end{figure}

     \begin{figure}[H]
            \centering
            \begin{subfigure}{.45\textwidth}
                \centering
                \includegraphics[width=\linewidth]{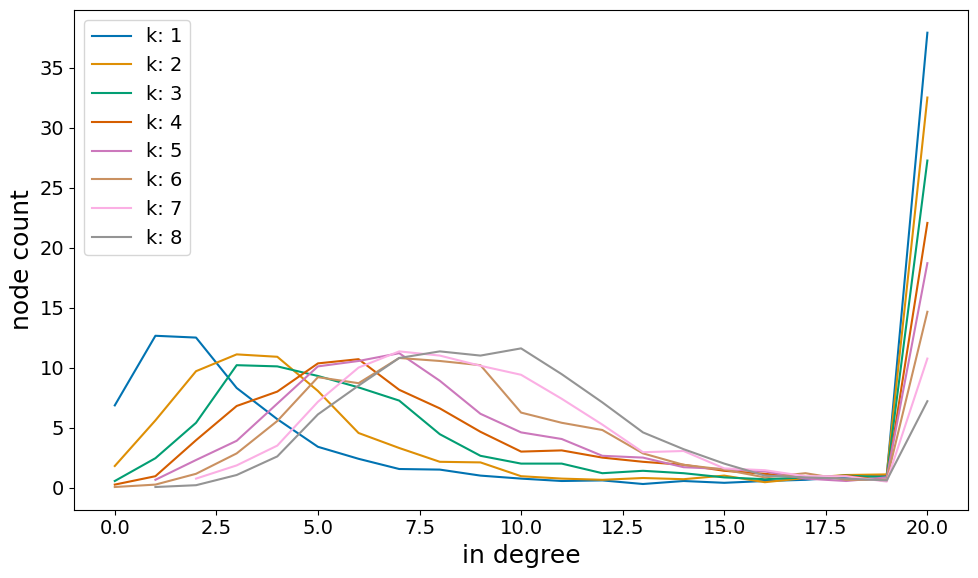}
                \caption{capped}
            \end{subfigure}
            \begin{subfigure}{.45\textwidth}
                \centering
                \includegraphics[width=\linewidth]{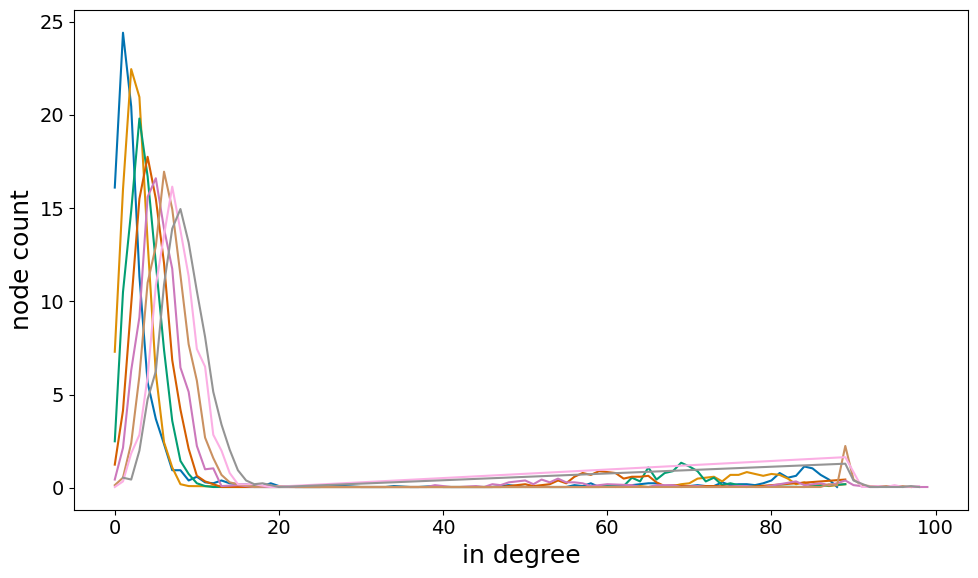}
                \caption{uncapped}
            \end{subfigure}
            \caption{In-degree distribution for different $k$-worst values for sub-network of miners. The case when all nodes are miners looks the same. For the capped case, the larger the $k$, the closer the network behaves to the initial random graph.}
            \label{fig:din-ks}
    \end{figure}

\section{Discussion}
In this work, we study the impact of greedy peering choices in a P2P network. In our models, we consider some nodes (either all or a subset of nodes) as ``miners'' as the source of messages that all nodes are trying to optimize their distance to. We mostly consider the case of homogeneous miners\footnote{We explore miners having different weights (representing different mining powers) in Appendix~\ref{app:miner-heterogeneity} and find that the more disparate the weights, the more the network behaves like the subset miner network. }, so it is often the case that a nodes' peers will tie in their scores. In fact, in our theoretical results, the choice of tie-breaking rule is a determinant in whether the network will stabilize. Our choice of edges having a weight of 1 symbolizes the delay in propagation each hop adds, normalized over some period of time. In reality, this weight would be an average of delays, and therefore ties between peers would be quite rare. We consider our model of lowering hop distance (as opposed to link-latency distance) for a network where link latency has less of an impact than the processing latency added by each node. It seems that the global-ordering rule fits best as a stand-in for this, assuming the processing latency of any two nodes don't differ too much. Node heterogeneity of this kind and in default parameters is left for future work.


%
Both our theoretical results and simulations look at the conditions for stability and how the network evolves into a possibly stable state. In both cases, we see a two-tier system for the miners in a stable topology and an additional third tier of all other nodes. An interesting direction for future work is what happens when new nodes (with or without the preferential property) join these stable networks: Does the network (re)stabilize and how long does it take?  

%
In both our simulations and theoretical analysis, stability properties are observed in the regime of a small $m$ relative to $n$, we note that cryptocurrency networks historically fall under this regime with a minority of network participants being miners \cite{kiffer2021under}. 
In our simulations, we explore how networks stabilize, particularly for network conditions that are difficult to theorize about.  While we explore several parameter choices, one caveat of our simulations is that we stay within relatively small network sizes (compared to the major existing cryptocurrencies, e.g., Bitcoin and Ethereum) so that the computation is tractable and we are still able to observe network dynamics\footnote{We note, however, that our simulations are within the realm of many smaller cryptocurrency networks such as the Ethereum Classic Network with around 500 observed nodes \cite{etcnodes}.}. In the same spirit, we consider a single $d$ value while changing the network size and $d_{in}$ caps to explore the proportional effect of the out-degree. Our theoretical results, however, deal with more general parameters. 


    \bibliographystyle{plain}
    \bibliography{main.bib}

\begin{thebibliography}{10}

\bibitem{albert2002statistical}
R{\'e}ka Albert and Albert-L{\'a}szl{\'o} Barab{\'a}si.
\newblock Statistical mechanics of complex networks.
\newblock {\em Reviews of modern physics}, 74(1):47, 2002.

\bibitem{aradhya2022overchain}
Vijeth Aradhya, Seth Gilbert, and Aquinas Hobor.
\newblock Overchain: Building a robust overlay with a blockchain.
\newblock {\em arXiv preprint arXiv:2201.12809}, 2022.

\bibitem{babel2022strategic}
Kushal Babel and Lucas Baker.
\newblock Strategic peer selection using transaction value and latency.
\newblock In {\em Proceedings of the 2022 ACM CCS Workshop on Decentralized Finance and Security}, pages 9--14, 2022.

\bibitem{BalaGoyal}
Venkatesh Bala and Sanjeev Goyal.
\newblock A noncooperative model of network formation.
\newblock {\em Econometrica}, 68(5):1181--1229, 2000.

\bibitem{bitcoincore}
{Bitcoin} core 24.0.1.
\newblock \url{https://github.com/bitcoin/bitcoin}, 2023.

\bibitem{coinmarketcap}
Today's cryptocurrency prices by market cap.
\newblock \url{https://coinmarketcap.com/}, 2023.

\bibitem{corbo05}
Jacomo Corbo and David~C. Parkes.
\newblock The price of selfish behavior in bilateral network formation.
\newblock In {\em Proc. of PODC'05}, pages 99--107, 2005.

\bibitem{delgado2019txprobe}
Sergi Delgado-Segura, Surya Bakshi, Cristina P{\'e}rez-Sol{\`a}, James Litton, Andrew Pachulski, Andrew Miller, and Bobby Bhattacharjee.
\newblock Txprobe: Discovering {Bitcoin}’s network topology using orphan transactions.
\newblock In {\em Financial Cryptography and Data Security: 23rd International Conference, FC 2019, Frigate Bay, St. Kitts and Nevis, February 18--22, 2019, Revised Selected Papers 23}, pages 550--566. Springer, 2019.

\bibitem{demaine+hmz:network}
Erik~D. Demaine, Mohammadtaghi Hajiaghayi, Hamid Mahini, and Morteza Zadimoghaddam.
\newblock The price of anarchy in network creation games.
\newblock {\em ACM Trans. Algorithms}, 8(2):1--13, 2012.

\bibitem{libp2p}
devp2p.
\newblock \url{https://github.com/ethereum/devp2p}, 2023.

\bibitem{dotan2021survey}
Maya Dotan, Yvonne-Anne Pignolet, Stefan Schmid, Saar Tochner, and Aviv Zohar.
\newblock Survey on blockchain networking: Context, state-of-the-art, challenges.
\newblock {\em ACM Computing Surveys (CSUR)}, 54(5):1--34, 2021.

\bibitem{gouel2021ip}
Matthieu Gouel, Kevin Vermeulen, Olivier Fourmaux, Timur Friedman, and Robert Beverly.
\newblock {IP} geolocation database stability and implications for network research.
\newblock In {\em Network Traffic Measurement and Analysis Conference}, 2021.

\bibitem{heilman2015eclipse}
Ethan Heilman, Alison Kendler, Aviv Zohar, and Sharon Goldberg.
\newblock Eclipse attacks on {Bitcoin}’s peer-to-peer network.
\newblock In {\em 24th $\{$USENIX$\}$ Security Symposium ($\{$USENIX$\}$ Security 15)}, pages 129--144, 2015.

\bibitem{henningsen2019eclipsing}
Sebastian Henningsen, Daniel Teunis, Martin Florian, and Bj{\"o}rn Scheuermann.
\newblock Eclipsing {Ethereum} peers with false friends.
\newblock {\em arXiv preprint arXiv:1908.10141}, 2019.

\bibitem{jackson+w:strategic}
Matthew Jackson and Asher Wolinsky.
\newblock A strategic model of social and economic networks.
\newblock {\em Journal of Economic Theory}, 71:44--74, 1996.

\bibitem{kiffer2023security}
Lucianna Kiffer, Joachim Neu, Srivatsan Sridhar, Aviv Zohar, and David Tse.
\newblock Security of blockchains at capacity.
\newblock {\em arXiv preprint arXiv:2303.09113}, 2023.

\bibitem{kiffer2021under}
Lucianna Kiffer, Asad Salman, Dave Levin, Alan Mislove, and Cristina Nita-Rotaru.
\newblock Under the hood of the {Ethereum} gossip protocol.
\newblock In {\em International Conference on Financial Cryptography and Data Security}, pages 437--456. Springer, 2021.

\bibitem{LAOUTARIS20141266}
Nikolaos Laoutaris, Laura Poplawski, Rajmohan Rajaraman, Ravi Sundaram, and Shang-Hua Teng.
\newblock Bounded budget connection ({BBC}) games or how to make friends and influence people, on a budget.
\newblock {\em Journal of Computer and System Sciences}, 80(7):1266--1284, 2014.

\bibitem{li2021toposhot}
Kai Li, Yuzhe Tang, Jiaqi Chen, Yibo Wang, and Xianghong Liu.
\newblock Toposhot: uncovering {Ethereum}'s network topology leveraging replacement transactions.
\newblock In {\em Proceedings of the 21st ACM Internet Measurement Conference}, pages 302--319, 2021.

\bibitem{Ethereum-lighthouse}
Lighthouse book: Advanced networking.
\newblock \url{https://lighthouse-book.sigmaprime.io/advanced_networking.html}.

\bibitem{lumezanu2009triangle}
Cristian Lumezanu, Randy Baden, Neil Spring, and Bobby Bhattacharjee.
\newblock Triangle inequality variations in the internet.
\newblock In {\em Proceedings of the 9th ACM SIGCOMM conference on Internet measurement}, pages 177--183, 2009.

\bibitem{mao2020perigee}
Yifan Mao, Soubhik Deb, Shaileshh~Bojja Venkatakrishnan, Sreeram Kannan, and Kannan Srinivasan.
\newblock Perigee: Efficient peer-to-peer network design for blockchains.
\newblock In {\em Proceedings of the 39th Symposium on Principles of Distributed Computing}, pages 428--437, 2020.

\bibitem{marcus2018low}
Yuval Marcus, Ethan Heilman, and Sharon Goldberg.
\newblock Low-resource eclipse attacks on {Ethereum}'s peer-to-peer network.
\newblock {\em Cryptology ePrint Archive}, 2018.

\bibitem{meirom+mo:network}
Eli~A. Meirom, Shie Mannor, and Ariel Orda.
\newblock Network formation games with heterogeneous players and the internet structure.
\newblock In {\em Proceedings of the Fifteenth ACM Conference on Economics and Computation}, EC '14, page 735–752, New York, NY, USA, 2014. Association for Computing Machinery.

\bibitem{miller2015discovering}
Andrew Miller, James Litton, Andrew Pachulski, Neal Gupta, Dave Levin, Neil Spring, and Bobby Bhattacharjee.
\newblock Discovering {Bitcoin}’s public topology and influential nodes.
\newblock {\em et al}, 2015.

\bibitem{bwlimitedposlc}
Joachim Neu, Srivatsan Sridhar, Lei Yang, David Tse, and Mohammad Alizadeh.
\newblock Longest chain consensus under bandwidth constraint.
\newblock In {\em {AFT}}. {ACM}, 2022.

\bibitem{etcnodes}
Etc node explorer.
\newblock \url{https://etcnodes.org/}.
\newblock Accessed on 08-08-2023.

\bibitem{osborne+r:game}
M.J. Osborne and A.~Rubinstein.
\newblock {\em A Course in Game Theory}.
\newblock MIT Press, 1994.

\bibitem{park2019nodes}
Sehyun Park, Seongwon Im, Youhwan Seol, and Jeongyeup Paek.
\newblock Nodes in the {Bitcoin} network: Comparative measurement study and survey.
\newblock {\em IEEE ACCESS}, 7:57009--57022, 2019.

\bibitem{poese2011ip}
Ingmar Poese, Steve Uhlig, Mohamed~Ali Kaafar, Benoit Donnet, and Bamba Gueye.
\newblock {IP} geolocation databases: Unreliable?
\newblock {\em ACM SIGCOMM Computer Communication Review}, 41(2):53--56, 2011.

\bibitem{rohrer2019kadcast}
Elias Rohrer and Florian Tschorsch.
\newblock Kadcast: A structured approach to broadcast in blockchain networks.
\newblock In {\em Proceedings of the 1st ACM Conference on Advances in Financial Technologies}, pages 199--213, 2019.

\bibitem{tang2022strategic}
Weizhao Tang, Lucianna Kiffer, Giulia Fanti, and Ari Juels.
\newblock Strategic latency reduction in blockchain peer-to-peer networks.
\newblock {\em arXiv preprint arXiv:2205.06837}, 2022.

\bibitem{toshniwal2021comparative}
Bhavesh Toshniwal and Kotaro Kataoka.
\newblock Comparative performance analysis of underlying network topologies for blockchain.
\newblock In {\em 2021 International Conference on Information Networking (ICOIN)}, pages 367--372. IEEE, 2021.

\bibitem{watts1998collective}
Duncan~J Watts and Steven~H Strogatz.
\newblock Collective dynamics of ‘small-world’ networks.
\newblock {\em Nature}, 393(6684):440--442, 1998.

\bibitem{zich2008jumpnet}
Jan Zich, Yoshiharu Kohayakawa, Vojtech R{\"o}dl, and Vaidy Sunderam.
\newblock Jumpnet: Improving connectivity and robustness in unstructured {P2P} networks by randomness.
\newblock {\em Internet Mathematics}, 5(3):227--250, 2008.

\end{thebibliography}

    \appendix
    \section{Additional Simulations}
        We also briefly consider other experimental setups, namely the starting topology of our simulations and heterogeneity in the weight of individual miners in the scoring function.

        \subsection{Alternative initial graph topology} 
        \label{app:alt-graph}

        \begin{figure} [H]
                \centering
                \begin{subfigure}{.32\textwidth}
                    \centering
                    \includegraphics[width=\linewidth]{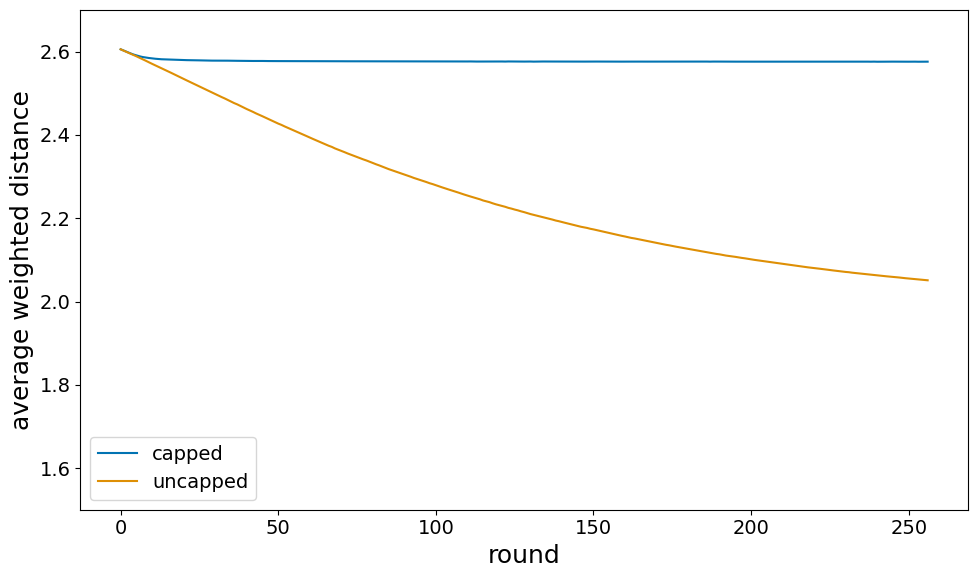}
                    \caption{Random (homogeneous)}
                \end{subfigure}
                \begin{subfigure}{.32\textwidth}
                    \centering
                    \includegraphics[width=\linewidth]{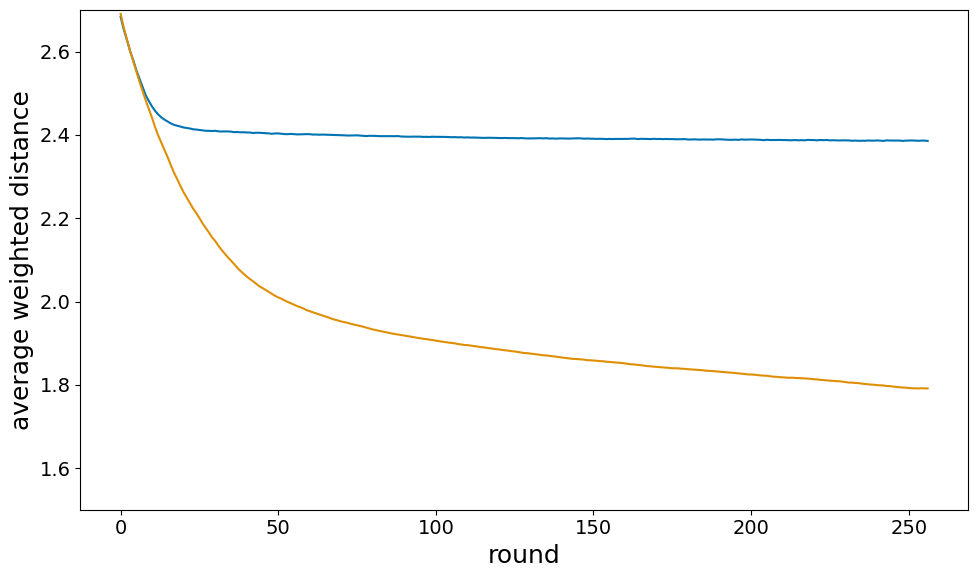}
                    \caption{Small world (homogeneous)}
                \end{subfigure}
                \begin{subfigure}{.32\textwidth}
                    \centering
                    \includegraphics[width=\linewidth]{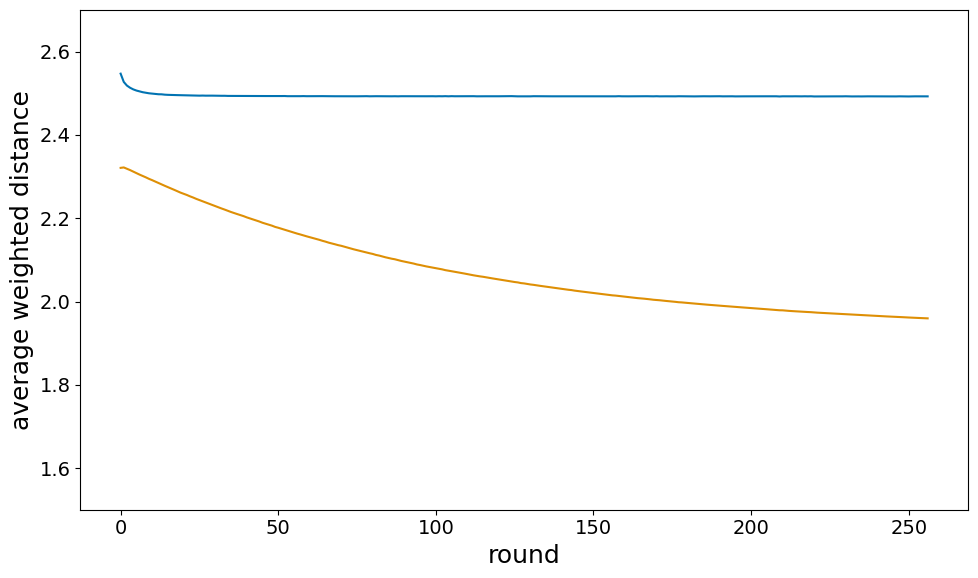}
                    \caption{Scale-free (homogeneous)}
                \end{subfigure}\\
                \begin{subfigure}{.32\textwidth}
                    \centering
                    \includegraphics[width=\linewidth]{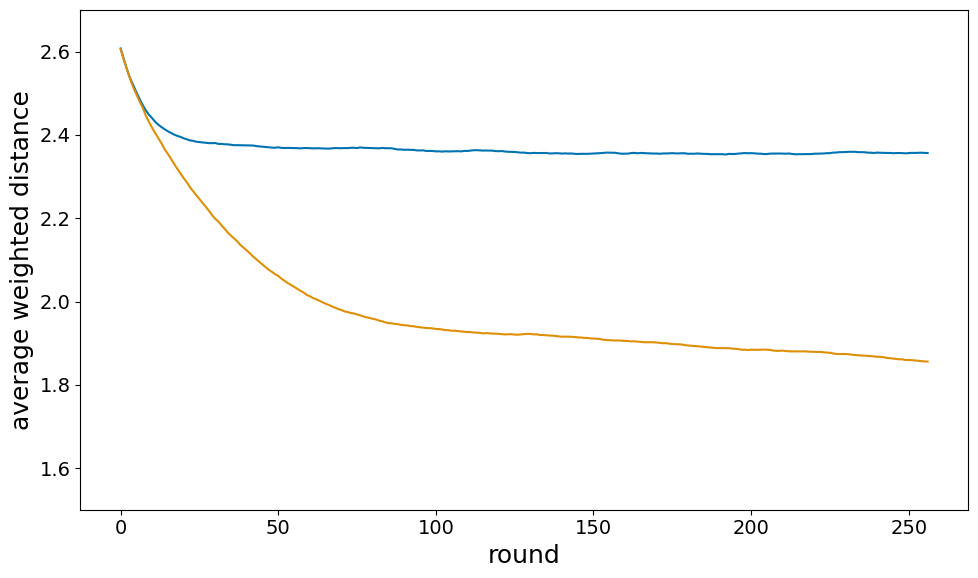}
                    \caption{Random (10  miners)}
                \end{subfigure}
                \begin{subfigure}{.32\textwidth}
                    \centering
                    \includegraphics[width=\linewidth]{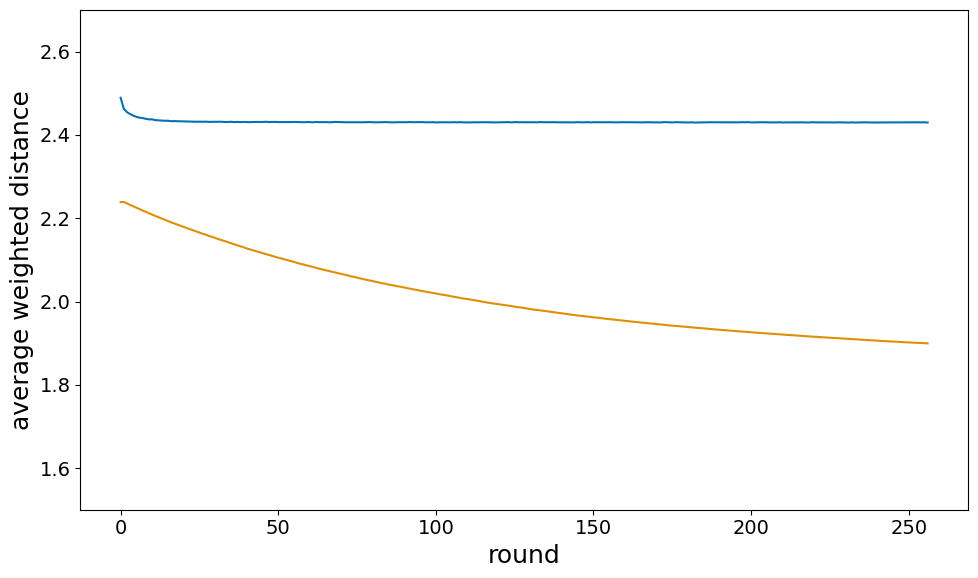}
                    \caption{Scale-free (10 miners)}
                \end{subfigure}
                \caption{Average distance to miners per round for 900 node networks starting from a random, small-world, or scale-free topology for networks of all homogeneous miners or a subset of 10 miners.}
                \label{fig:other-net-avg-dist}
        \end{figure} 
        
        We look at the impact starting from a random topology has on our results by comparing two other models for initial graphs topologies. We consider a small-world initial topology based on the Watts-Strogratz model with probability 0.5 and a scale-free network based on the Barabasi-Albert model with an initial connected component of size 20. For both graphs, we look at a network of all homogeneous miners. For the scale-free network we also consider the case with only 10 miners who all start in the initial connected component\footnote{We don't consider the subset miner case for the small-world graph as it is not clear where to put the miners in the initial ring lattice, which could lead to large biases in our results.}. 
        
        We see the average distance to miners for these three simulations with $n=900$ in Figure~\ref{fig:other-net-avg-dist} may start at different values, but converge around the same point. \textbf{Generally, the small-world simulations behave very similarly to the initial random network simulations overall.} The only exception is that though the average distance to miners has similar behavior, how individual nodes experience this distance is very different. This is seen in the PDF of the distance to miners in Figure~\ref{fig:small-world-dist} in with the small-world network converging to more nodes having smaller distance. Note the step-wise behavior in the small world network is an artifact of the initial ring structure of the Watts-Strogratz network.

        The scale-free network, on the other hand is particularly interesting in the capped and uncapped cases. With capped $d_{in}$, the average distance to miners (both in the homogeneous and 10 miner case) starts lower than the other networks and stays practically at the same value. Interestingly, with only a subset of the network as miners,  when we compare the eccentricity and diameter of the scale-free versus random network simulation in Figure~\ref{fig:other-diam-eccen} we see the scale-free diameters remain somewhat constant and at higher levels than the random graph for the miners. This suggests that a scale-free network is somewhat stable under \greedy but that the network we arrive from a random network is different.
    
        \textbf{The main take away is that generally it appears that \greedy is converging to a network with scale-free properties, but if we start from such a network we get instances where miners are particularly worse off in the network- suggesting the stability of the initial configuration leads to a different evolution of the network.}

        \begin{figure} [t]
                \centering
                \begin{subfigure}{.44\textwidth}
                    \centering
                    \includegraphics[width=\linewidth]{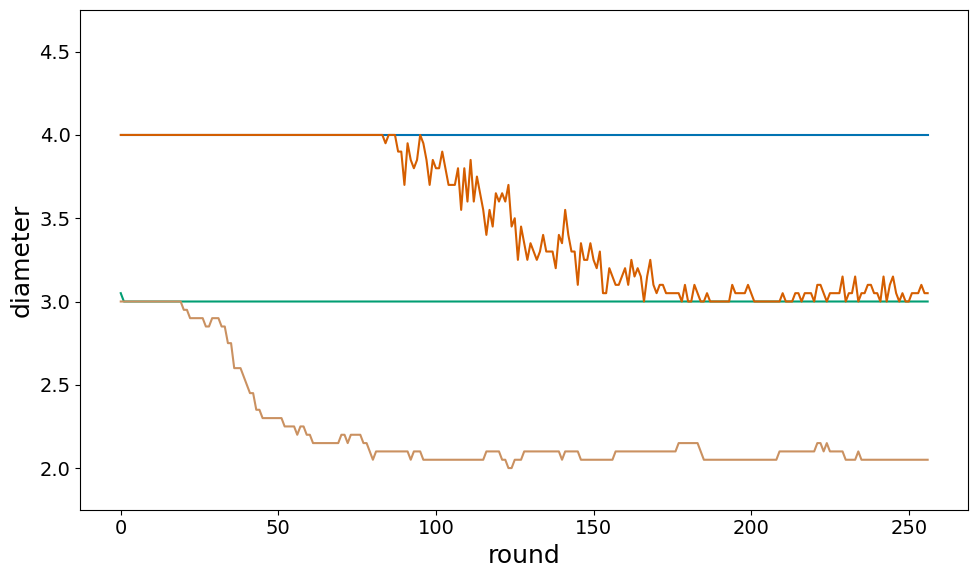}
                    \caption{random network diameter}
                \end{subfigure}
                \begin{subfigure}{.44\textwidth}
                    \centering
                    \includegraphics[width=\linewidth]{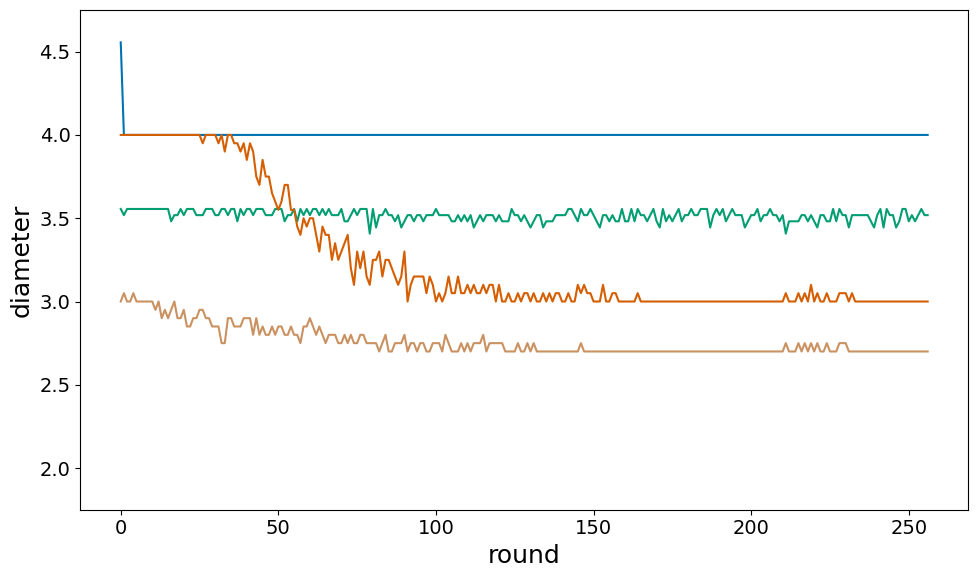}
                    \caption{scale-free diameter}
                \end{subfigure}\\
                \begin{subfigure}{.44\textwidth}
                    \centering
                    \includegraphics[width=\linewidth]{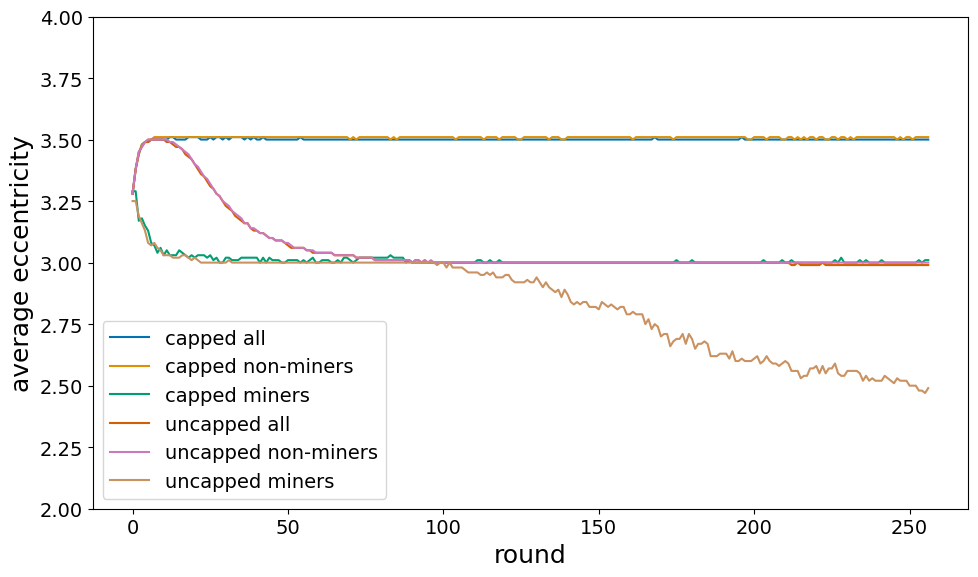}
                    \caption{random network eccentricity}
                \end{subfigure}
                \begin{subfigure}{.44\textwidth}
                    \centering
                    \includegraphics[width=\linewidth]{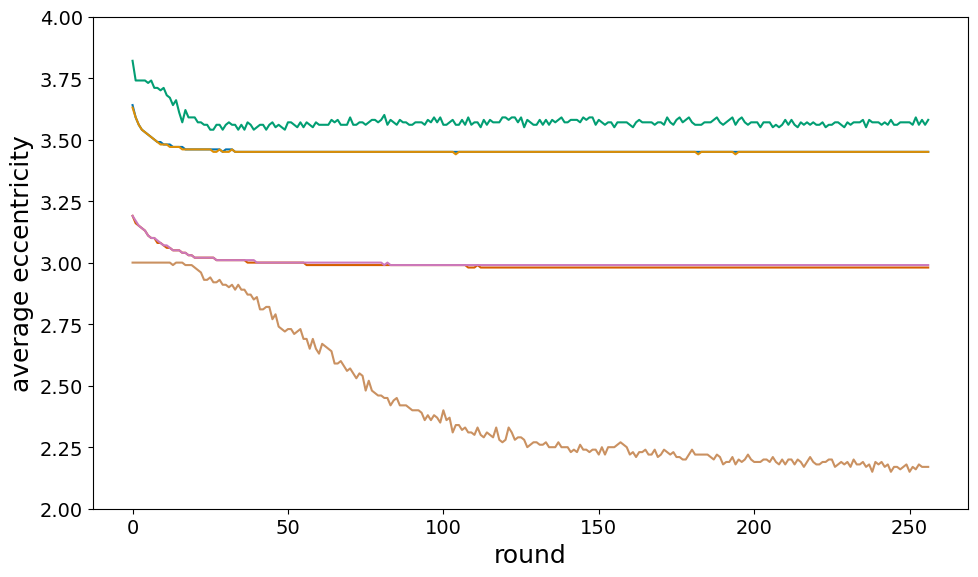}
                    \caption{scale-free eccentricity}
                \end{subfigure}
                \caption{Diameter and eccentricity over time for simulations of 900 nodes and 10 miners starting from a random or a scale-free topology.}
                \label{fig:other-diam-eccen}
        \end{figure}

        \begin{figure} [H]
                \centering
                \begin{subfigure}{.32\textwidth}
                    \centering
                    \includegraphics[width=\linewidth]{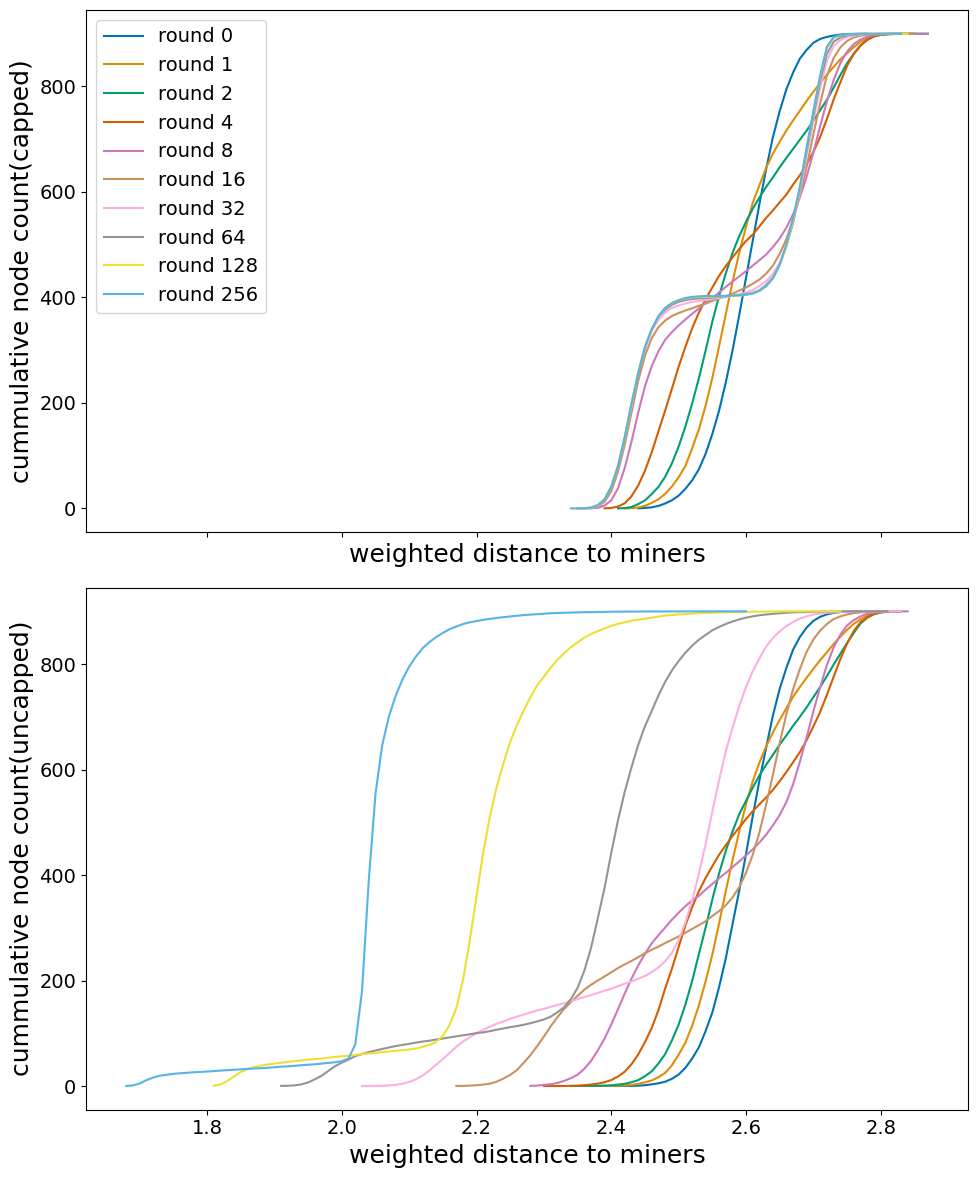}
                    \caption{Random}
                \end{subfigure}
                \begin{subfigure}{.32\textwidth}
                    \centering
                    \includegraphics[width=\linewidth]{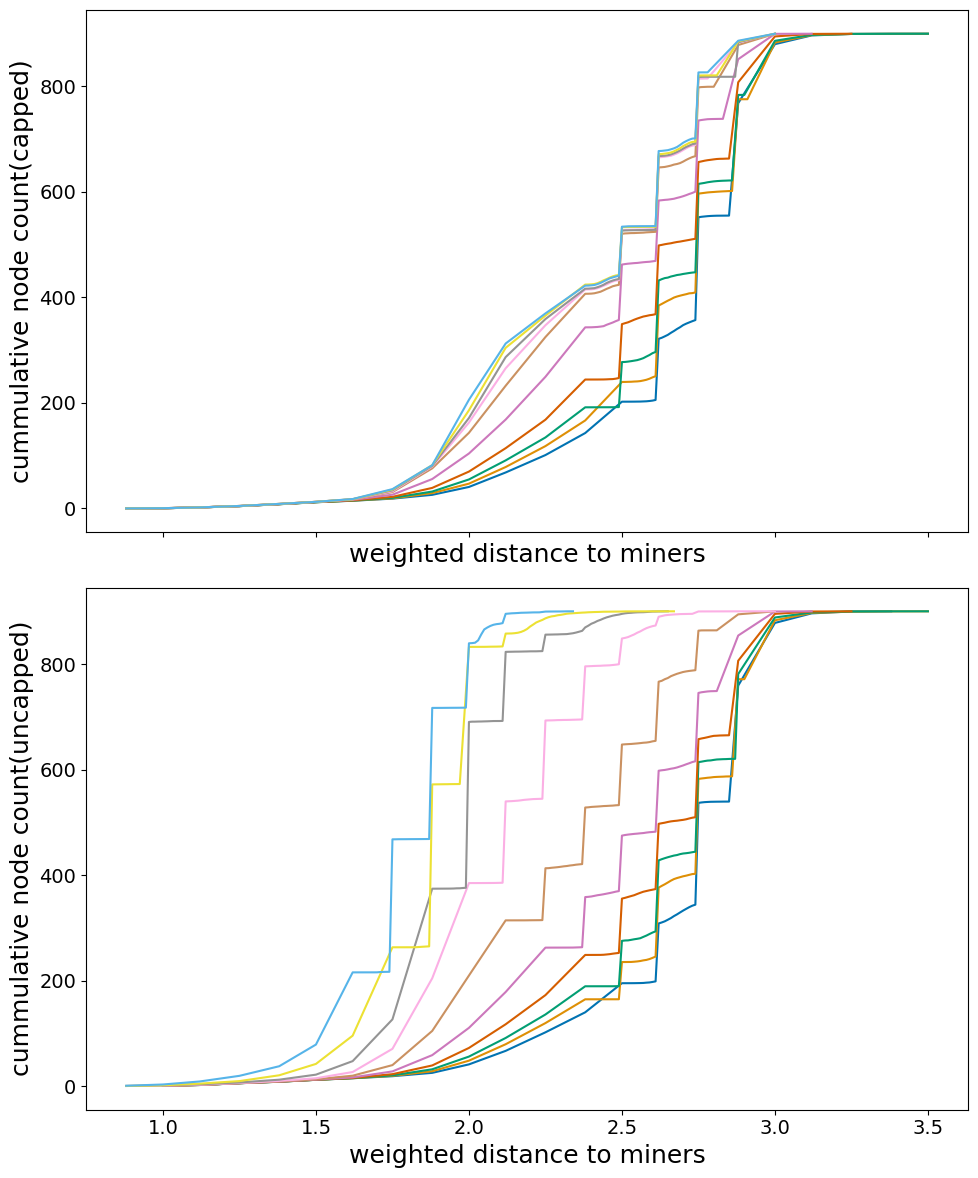}
                    \caption{Small-world}
                \end{subfigure}
                \begin{subfigure}{.32\textwidth}
                    \centering
                    \includegraphics[width=\linewidth]{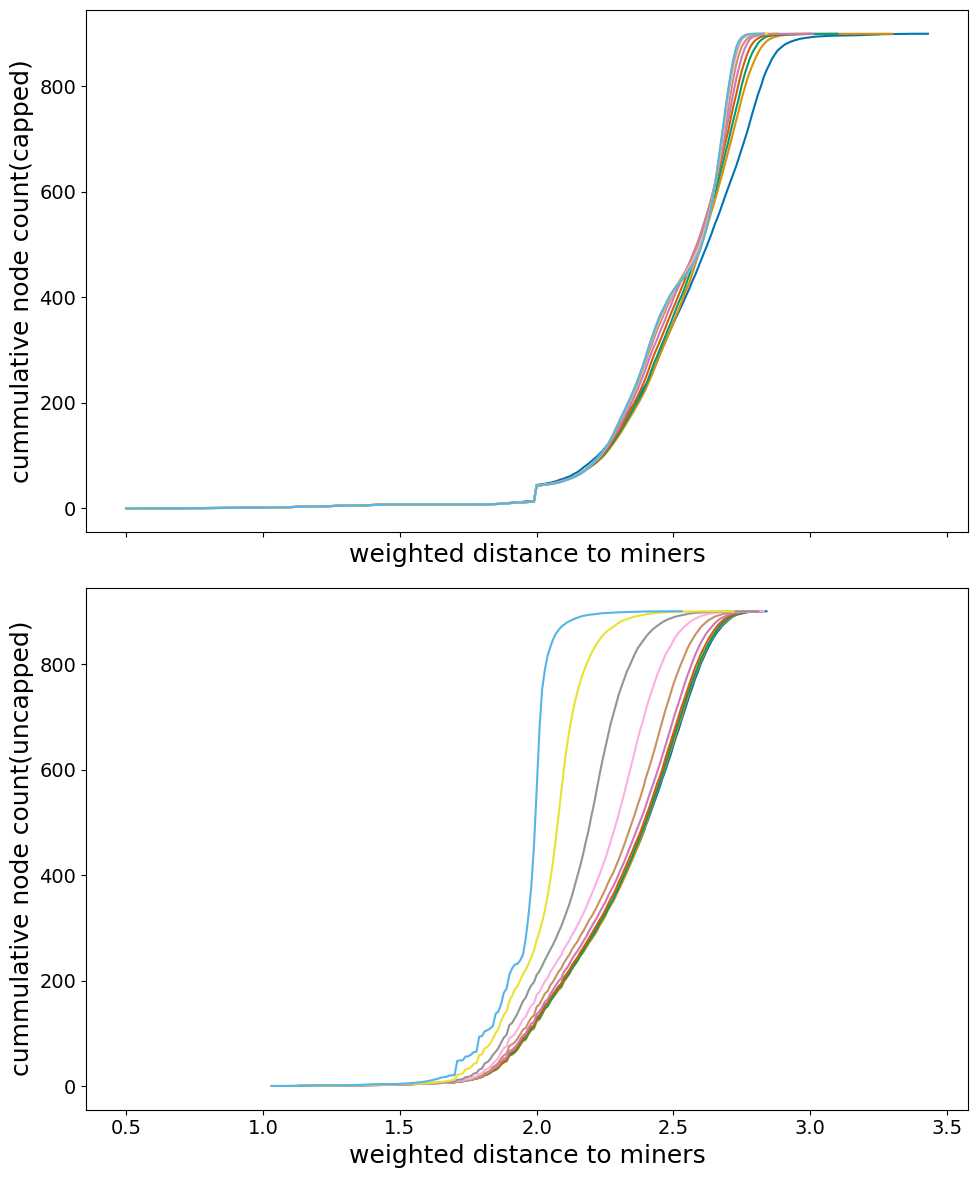}
                    \caption{Scale-free}
                \end{subfigure}
                \caption{PDF of the average distance to miners for simulations starting from a random graph and small world graph of homogeneous nodes.}
                \label{fig:small-world-dist}
        \end{figure}

        \subsection{Miner Heterogeneity}
        \label{app:miner-heterogeneity}
        In this work, we primarily consider miners as homogeneous entities. We consider also a couple examples of heterogeneous hashrate distributions among the miners, where node scores are weighted by the fraction of the hashrate each miner controls. We consider miners whose hashrate are linearly diminishing (miner $i$ has hashrate $\frac{i}{i+1}$ ), and miners whose hashrate are exponentially diminishing (miner $i$ has hashrate $\frac{1}{2}^i$). For these runs we look at 100 nodes with $d=10$ and $d_{in}= 20$ or uncapped. \textbf{Our main observation is that the bigger the disparity in hashrates, the more the network acts like the case where only a small subset of the network are miners.} We see this in the average distance to miners in Figure~\ref{fig:avg-d-hashrates}. The linearly diminishing hashrate (right) behaves very similarly to the homogeneous miner case. The exponentially decreasing hashrate, however, behaves different. Most obviously, we observe the range on the y-axis and that the exponential case has an order of magnitude greater decrease in average distance to miners (i.e., to each of the other nodes). We observe also that the more disparity there is in mining hashrates, the more the network behaves like the subset miner case.
        
        \begin{figure}[H]
                \centering
                \begin{subfigure}{.45\textwidth}
                    \centering
                    \includegraphics[width=\linewidth]{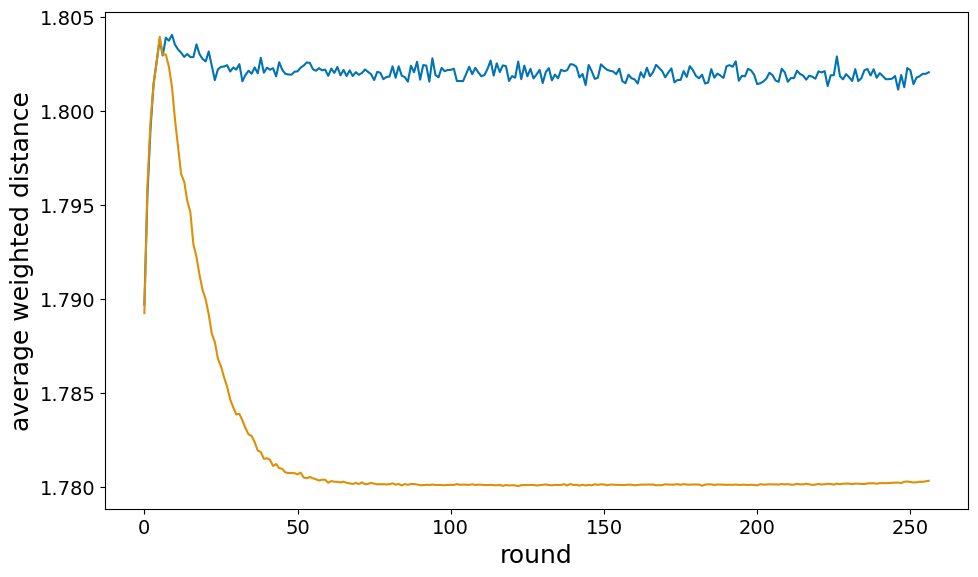}
                    \caption{Linear diminishing hashrate}
                \end{subfigure}
                \begin{subfigure}{.45\textwidth}
                    \centering
                    \includegraphics[width=\linewidth]{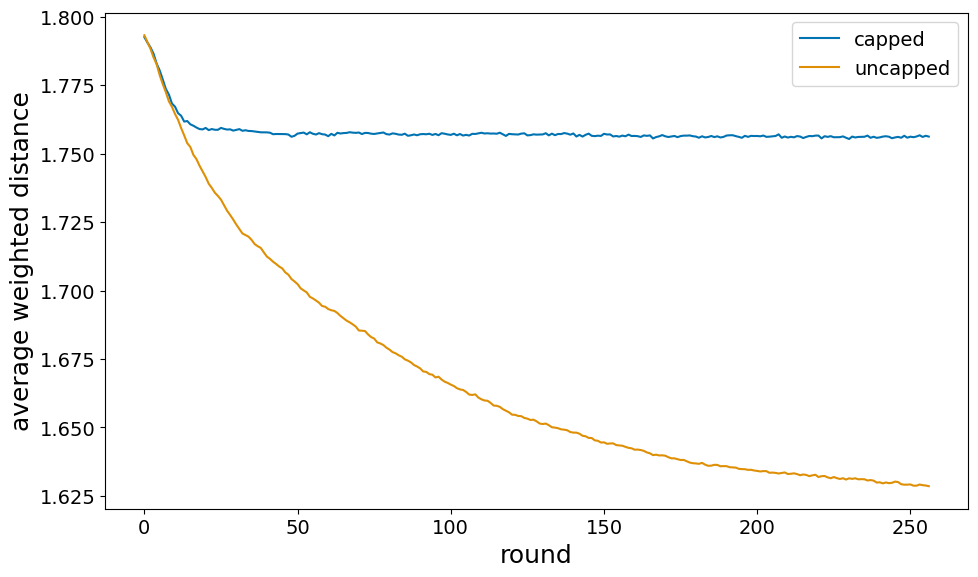}
                    \caption{Exponential diminishing hashrate}
                \end{subfigure}
                \caption{Average weighted (by hashrate) distance to miners per round for an 100 node network with all miners and hashrate distributions linearly and exponentially decreasing.}
                \label{fig:avg-d-hashrates}
        \end{figure}  

\end{document}